\documentclass[a4paper, 11pt]{article}

\usepackage{a4wide}
\usepackage[normalem]{ulem}        
\usepackage{fancyhdr,fancybox}
\usepackage{amsfonts, amssymb, amsmath, ulem, amsthm, amstext, fancybox}
\usepackage{enumerate}
\usepackage{caption,graphicx,wrapfig,subfig}
\usepackage{stmaryrd}
\usepackage{color}
\usepackage{cite}
\usepackage{epstopdf}

\usepackage{defs}
\usepackage{fonts}

\newcommand{\eps}{\varepsilon}

\newcommand{\seqalign}[2]{\begin{subequations} #1 \begin{align} #2 \end{align}\end{subequations}}
\newcommand{\balign}[1]{\begin{align} #1 \end{align}}

\newtheoremstyle{def}
	{0.5cm}                   
  {0.5cm}                   
  {}           							
  {}                      
  {\bfseries}  					  
  {}                      
  {\newline}        			
  {\underline{\thmname{#1} \thmnumber{#2:}} \thmnote{[#3]}}
  {}                       

\newcommand{\q}{\quad}
\newcommand{\tx}{\text}

\newcommand{\ds}{\displaystyle}

\newcommand{\reffig}[1]{\text{Figure~}\ref{#1}}

\newtheorem{prop}{Proposition}
\newtheorem{lem}{Lemma}
\newtheorem{rem}{Remark}
\newtheorem{theorem}{Theorem}

\def\Id{{\rm{Id}}}
\def\ptm{$P\!M_1$ }

\def\psit{\tilde{\psi}}
\def\psic{\psit^{\rm{c}}}
\def\psid{\psit^{\rm{d}}}

\def\P{\cP_{\bu}}
\def\Pt{\tilde{\cP}_{\bu}}
\def\Q{\cQ_{\bu}}
\def\Qt{\tilde{\cQ}_{\bu}}

\def\ahat{\hat{\bsalpha}}
\def\etad{\eta_*}

\def\fe{\bff^{\cE}}
\def\fc{\bff^{\rm{C}}}
\def\fd{\bff^{\rm{D}}}

\def\fdh{\hat{\bff}^{\rm{D}}}

\def\pp{\Pi^{\rm{P_1}}}
\def\pe{\Pi^{\rm{M_1}}}
\def\pc{\Pi^{\rm{C}}}
\def\pd{\Pi^{\rm{D}}}
\def\pdel{\Pi_{\delta}}
\def\pdell{\Pi_{\delta,\ell}}

\def\qe{Q^{^{\rm{M_1}}}}

\def\E{\cE(\bu)}

\def\div{\grad_x \cdot}

\def\sigt{\sigma_{\rm{t}}}
\def\siga{\sigma_{\rm{a}}}
\def\sigs{\sigma_{\rm{s}}}
\def\rs{r_{\rm{s}}}
\def\ra{r_{\rm{a}}}

\def\wbar{\overline}

\def\xjalpha{{\hat{x}}_j^\ell}

\begin{document}
\title{Perturbed, Entropy-Based Closure for Radiative Transfer}

\author{Martin Frank $^{\rm a}$\thanks{Corresponding author. Email: frank@mathcces.rwth-aachen.de} , Cory D. Hauck $^{\rm b}$\thanks{The research of
this author is sponsored by the Office of
Advanced Scientific Computing Research; U.S. Department of Energy. The work was
performed at the Oak Ridge National Laboratory, which is managed by UT-Battelle,
LLC under Contract No. De-AC05-00OR22725. Accordingly, the U.S. Government
retains a non-exclusive, royalty-free license to publish or reproduce the
published form of this contribution, or allow others to do so, for U.S.
Government purposes. \vspace{6pt}} , Edgar Olbrant $^{\rm a}$ , \\ $^{\rm a}$ \small{\textit{RWTH Aachen University, Department of Mathematics \& Center for Computational}} \\ \small{\textit{Engineering Science, Schinkelstrasse 2, D-52062 Aachen, Germany}} \\
$^{\rm b}$ \small{\textit{Oak Ridge National Laboratory, 1 Bethel
Valley Road, Oak Ridge, Tennessee, 37831 USA}}
}

\date{\today}
\maketitle

\begin{abstract}
We derive a hierarchy of closures based on perturbations of well-known entropy-based closures;  
we therefore refer to them as \textit{perturbed entropy-based models}. Our derivation reveals final 
equations containing an additional convective and diffusive term which are added to the 
flux term of the standard closure. We present numerical simulations for the simplest member of the hierarchy, the
\textit{perturbed $M_1$} or \textit{\ptm model}, in one spatial dimension. Simulations are performed using a Runge-Kutta 
discontinuous Galerkin method with special limiters that
guarantee the realizability of the moment variables and the positivity of the material
temperature. Improvements to the 
standard $M_1$ model are observed in cases where unphysical shocks develop in the $M_1$ model.
\end{abstract}

\section{Introduction}
In this paper, we derive a new hierarchy of kinetic moment models in the
context of frequency-integrated (grey) photon transport.  These new models
are perturbations of well known entropy-based models;  we therefore refer to them
as \textit{perturbed entropy-based} or \textit{PEB models}.  We
present numerical simulations for the simplest member of the hierarchy, the
\textit{perturbed $M_1$} or \textit{\ptm model}, in one spatial dimension.  In
this setting, the \ptm model approximates the evolution of the photon radiation
energy $E$ and radiation flux $F$ through a material medium with slab geometry.
The photons interact with the material through scattering and
emission/absorption processes.


Entropy-based (EB) models have been studied extensively in areas such as
extended
thermodynamics \cite{Muller-Ruggeri-1993,Dreyer-1987}, gas dynamics
\cite{Levermore-1996, Hauck-Levermore-Tits-2008,Junk-1998,Junk-2000,
Schneider-2004, Groth-MacDonald-2009,Levermore-Morokoff-Nadiga-1998},
semiconductors \cite{Anile-Romano-2000,Junk-Romano-2005,
Rosa-Mascali-Romano-2009,Anile-Muscato-1995,Anile-Allegretto-Ringhofer-1998,
Anile-Pennisi-1992, Levermore-1998,Jungel-Krause-Pietra-2007,Hauck-2006},
quantum fluids \cite{Degond-Ringhofer-2003,Dreyer_Hermann_Kunik_2004}, radiation
transport \cite{Hauck-McClarren-2010, Dubroca-Klar-2002,Brunner-Holloway-2001,
Brunner-Holloway-2005,Wright-Frank-Klar-2009, Monreal-Frank-2009,Hau10,
Dubroca-Feugas-1999,Minerbo-1978,Frank-Dubroca-Klar-2006,
Cernohorsky-Bludman-1994-2,Cernohorsky-vandenHorn-Cooperstein-1989,
Smit-vanderHor-Bludman-2000}, and phonon transport in solids
\cite{Dreyer_Hermann_Kunik_2004}. In the context of radiative transfer,
entropy-based
models are commonly referred to as \textit{$M_N$ models}, where $N$ is
order of the expansion.  The $M_1$ model dates back to
\cite{Minerbo-1978}, where it was first derived using Maxwell-Boltzmann
statistics.  For problems with Bose-Einstein statistics, formal theoretical
properties
such as hyperbolicity and entropy dissipation were first reported in
\cite{Dubroca-Feugas-1999} for arbitrary $N$. However, computational studies
have focused primarily on properties of the $M_1$
model and its extensions, including multigroup equations \cite{Tur05} and
partial moment models \cite{Dubroca-Klar-2002,Frank-Dubroca-Klar-2006}. In
related work, one may find simulations of $M_1$ models based on other
statistics, including Maxwell-Boltzmann
\cite{Brunner-Holloway-2001,BrunnerThesis, Brunner-Holloway-2005} and
Fermi-Dirac \cite{BluCer95,SmiCerBlu97}. This
attachment to $M_1$ is due to the fact that the higher order members of the
$M_N$ hierarchy require the repeated solution of expensive numerical
optimization problems.  However, simulations of the $M_2$ model
\cite{Wright-Frank-Klar-2009, Monreal-Frank-2009} (the next member in the
hierarchy) have been performed for Bose-Einstein statistics and for $M_N$ up to
order $N=15$ for special benchmark problems using Maxwell-Boltzmann statistics
\cite{Hau10, Alldredge-Hauck-Tits-2012}.

There are several reasons to consider perturbative modifications to EB models. 
First, it is more economical to improve the model with perturbative corrections
than to increase the number of moments, since the latter increases the memory
footprint and makes the defining optimization problem more difficult to solve.
In this case of grey photon transport, the flux in the $M_1$ model can be
expressed analytically, i.e., no direct solution of the optimization problem is
required.  Thus, in this case, the argument against increasing the number of
moments is
especially compelling.  A second reason is that perturbations add (among
other things) diffusive terms to the EB model.  It is hoped that these terms
will smooth out non-physical shocks which are known to exist in EB models. 
These shocks are a generic artifact of the modeling procedure that result when
approximating linear transport in phase space by a nonlinear hyperbolic balance
law for a set of moments.  A third reason is that the specification of boundary
conditions for moment equations that are consistent with the underlying
kinetic boundary conditions is an open problem.  However, at least in the case
of linear moment equations, recent efforts \cite{Levermore2009} have shown the
potential for well-posed boundary conditions for models with perturbative
corrections. A fourth and final reason is that entropy-based closure do not
depend on the properties of the material.  Perturbations on the other hand
can couple material properties into the closure.

Our goal in this work is to assess the qualitative behavior of the \ptm
model relative to the original $M_1$ model. We consider several test
cases and find that the \ptm model gives mixed results.  Roughly speaking, it
does quite well for shock problems that the $M_1$ model cannot handle.  However,
for more regular solutions, the two models perform comparably; and in some
cases, the $M_1$ model performs slightly better.



One of the fundamental questions associated with any moment model is the issue
of \textit{realizability}. 
In the context of the $M_1$ and $PM_1$ models, the two unknowns 
$E$ and $F$ are called realizable if and only if they are the first two moments of an
underlying kinetic distribution. This requirement on $E$ and $F$
is mathematically equivalent to the condition
\begin{equation}
\label{eq:real}
 | F |\leq cE \:,
\end{equation}
which must be satisfied point-wise in space and time. Here, $c$ is the speed of light.
It is expected that the solutions of the $M_1$ model will satisfy
\eqref{eq:real} because it
(like all EB models) is derived assuming an ansatz for the kinetic distribution
which is positive.  However, the underlying ansatz for the \ptm model is a
perturbation of the EB ansatz that is no necessarily positive. Therefore, a
modification of the PEB ansatz is needed which controls the contribution of the
perturbative term.

Even for the $M_1$ model, the realizability condition \eqref{eq:real} can be
destroyed by a numerical method unless special care is taken to enforce it. To
address this issue in the current setting, we build on previous work with the
$M_1$ model \cite{OlbHauFra11}, using a Runge-Kutta discontinuous Galerkin
(RKDG) method that is equipped with a special slope limiter \cite{ZhaShu10a,
ZhaShu10b} in the spatial variable. For implementation of the \ptm model, this
special limiter must be applied in combination with a control parameter that
limits the size of the perturbations in the underlying ansatz of the PEB
closure. The RKDG method \cite{BasReb97} is a natural discretization here
because we deal with a hyperbolic system of equations that is augmented by a
diffusive term.

The remainder of the paper is organized as follows.  In Section
\ref{sec:moments}, we introduce the radiative transfer equation and moment model
framework. In Section \ref{sec: closure}, we derive perturbed entropy-based
closures and give explicit expressions for the perturbed $M_1$ model.  In Section
\ref{sec:dg}, we focus on the $PM_1$ model in slab geometry and give details of
the discontinuous Galerkin method used for
simulation. In Section \ref{sec:results}, we present numerical results.  Section
\ref{sec:disc} is for discussion and conclusions.  Several calculations and
proofs are relegated to the Appendix.

\section{Radiative Transfer and Moment Equations} \label{sec:moments}
We consider a collection of photons which move at the speed of light $c$ through
a static material medium.  In engineering and physics applications, the
fundamental quantity of interest is the radiation intensity $\psi =
\psi(x,\Omega,\nu,t)$ which is a function of position $x \in K \subset \bbR^3$,
direction $\Omega \in \bbS^2$, frequency $\nu \in (0,\infty)$, and time $t \in
(0,\infty)$.  Roughly speaking, $\psi$ is the flux of energy through a surface
normal to $\Omega$.
If $f$ is the kinetic density of photons---that is, the number density with
respect to the Lebesgue measure $dx d\Omega d\nu$---then $\psi = h\nu cf$,
where $h$ is Planck's constant.

The material is characterized by a temperature $T=T(x)$, an equation of state
 for the energy $e=e(T)$, and by scattering, absorption, and total
cross-sections: $\sigs$, $\siga$, and $\sigt = \siga+\sigs$ that depend on $x$
directly and also indirectly through the material temperature.

\subsection{The Radiative Transfer Equation}
The radiative transfer equation, which approximates the evolution of $\psi$, is
given by
\begin{equation} \label{eq:transport}
\frac{1}{c} \p_t \psi +  \Omega \cdot \grad_x \psi = \cC(\psi;T) \:.
\end{equation}
The collision operator $\cC$ models interactions of photons with the
medium.  For the purposes of this paper, we assume $\cC$ has the form
\begin{equation} \label{eq:coll_op}
 \cC(\psi;T):= - \sigt\psi + \frac{1}{4 \pi} \left( \sigs \phi + \siga B(T)
+ s \right )
\:,
\end{equation}
where $\phi$ is the angular integral of $\psi$:
\begin{equation}
 \phi := \int_{\bbS^2} \psi d\Omega \:,
\end{equation}
and the Planckian
\begin{equation}
 B(T) := \frac{2h \nu^3}{c^2} 
    \frac{1}{\exp \left( \frac{h \nu } {kT}\right) - 1} \:
\end{equation}
models blackbody radiation from the material. The constant $k$ is
Boltzmann's constant.    The first term in $\cC$ accounts for the loss of
photons at a given frequency and angle due to both out-scattering and absorption
by the material.  The second group of terms gives the gain of photons due to
in-scattering from other angles, re-emission by the material, and a generic
external source $s$.  To make calculations, $s$ is assumed to be isotropic.
However, such an assumption is not necessary.

The evolution of the material temperature is determined by a balance of
emitting and absorbed photons:
\begin{equation}
\p_t e(T) = \siga \left( \vint{\psi} - acT^4 \right) \:,
\label{eq:material_energy}
\end{equation}
where angle brackets are used as a shorthand notation for integration
over
angle and frequency:
\begin{equation}
 \vint{\,\cdot\,} \equiv \int_0^\infty \int_{\bbS^2} (\,\cdot\,) \,d\Omega d\nu
\:,
\end{equation}
and the $T^4$ term in the first equation comes from the
Stefan-Boltzmann Law:
\begin{equation}
 \int_0^\infty B(T)d\nu = ac T^4 \,.
 \label{eq:SB_law}
\end{equation}
The constant $a$ is the \textit{radiation constant}. Though the material
equation \eqref{eq:material_energy} plays an important role, we will focus here
on simulating the transport equation \eqref{eq:transport}.

\subsection{Moment Equations}
The large phase space on which \eqref{eq:transport} is defined makes direct
numerical simulation prohibitively expensive.  Thus, approximate models are
needed to reduce the size of the system.  A common and well-known approach is
the method of moments, for which the angular and/or frequency dependency of
$\psi$ is approximated using a finite number of weighted averages.

Derivation of any moment system begins with the choice of a vector-valued
function $\bm:\bbS^2 \to \bbR^n, \, \Omega \mapsto
[m_0(\Omega),\ldots,m_{n-1}(\Omega)]^T$, whose $n$ components are linearly
independent functions of $\Omega$. Evolution equations for the moments
$\bu(x,t):=
\Vint{\bm \psi(x,\cdot,t)}$ are found by multiplying the transport equation by
$\bm$ and integrating over all angles to give
\begin{equation} \label{eq:moment_eqns}
 \frac{1}{c} \p_t \bu
   + \grad_x \cdot  \vint{\Omega \bm \psi}
   = \vint{\bm \cC(\psi; T)} \;.
 \end{equation}

The system \eqref{eq:moment_eqns} is not closed; a recipe,
or \textit{closure}, must be prescribed to express unknown quantities in terms
of the given moments. Often this is done via an
approximation for $\psi$ in \eqref{eq:moment_eqns} that depends on $\bu$,
 \begin{equation}
   \psi(x,\Omega,t) \simeq \cE(\bu(x,t))(\Omega) \:,
 \end{equation}
and satisfies the consistency relation
\begin{equation}\label{eq:consistency}
 \vint{\bm \cE(\bu)} = \bu \,.
\end{equation}
The resulting moment system is
\begin{equation} \label{eq:moment_eqns_closed}
 \frac{1}{c} \p_t \bu
   + \grad_x \cdot  \vint{\Omega \bm \E}
   = \vint{\bm \cC(\E; T)} \;.
 \end{equation}
In general, a closure is required to evaluate both the flux terms and the
collision terms in \eqref{eq:moment_eqns}.  However for the collision operator
in \eqref{eq:coll_op}, no closure is required.  Indeed, it is
straight-forward to show that 
$\vint{\bm \cC(\E;T)} = \vint{\bm \cC(\psi;T)}$ for any reconstruction that satisfies the consistency
relation.  Thus, for the purposes of this paper, we will be focused on closure
of the flux term.  As one might expect, the behavior of a moment
system---and in particular its ability to capture fundamental features of the
kinetic description---depends heavily on the form of the reconstruction.

\section{Entropy-Based and Perturbed Entropy-Based (PEB) Closures}
\label{sec: closure}
In this section, we briefly review the theory of entropy-based closures for
radiative transfer \cite{Dubroca-Klar-2002,Brunner-Holloway-2001,
Brunner-Holloway-2005,
Dubroca-Feugas-1999,Minerbo-1978,Frank-Dubroca-Klar-2006,Wright-Frank-Klar-2009,
Monreal-Frank-2009,Hau10} and introduce our new perturbative
model.

\subsection{Entropy-Based Closures}
A general strategy for prescribing a closure is to use the solution of a
constrained optimization problem
\begin{align} \label{eq:optimization}
 \min_{\footnotesize g \in \operatorname{Dom}(\cH)} ~&~ \cH(g) \\
 \mbox{s.t.} ~\quad&~ \vint{\bm g} = \vint{\bm \psi}
\end{align}
where $\cH(g) := \vint{\eta(g)}$ and $\eta: \bbR \rightarrow \bbR$ is a
strictly convex function that is related to the entropy of the system.  For
photons, the physically relevant entropy comes from Bose-Einstein statistics
and is given by \cite{Ore55,Ros54}
\begin{equation}
 \eta(g) = \frac{2k\nu^2}{c^3}\left[n_g\log(n_g)-(n_g+1)\log(n_g+1)\right] \:,
\end{equation}
where the $n_g$ is the occupation number associated with $g$:
\begin{equation}
 n_g := \frac{c^2}{2h\nu^3}g \:.
\end{equation}
The solution of \eqref{eq:optimization} is expressed in terms of 
the Legendre dual
\begin{equation}
	\etad(f) = - \frac{2k\nu^2}{c^3} \log\left(1 - \exp \left( - \frac{h\nu
c}{k }  f \right)\right) \:.
\end{equation}
 Let 
\begin{equation}
\cB(\bsalpha):= \etad'\left(\bsalpha^{T}\bm \right)=
\frac {2h \nu^3}{c^2}
    \frac{1}{\exp \left( - \frac{h\nu c}{k }  \bsalpha ^T \bm\right) - 1}
    \label{eq:entropy_ansatz} \:.
\end{equation}
Then we have the following.

\begin{theorem}
\label{thm:DubFeu}
The solution of \eqref{eq:optimization} is given by $\cB(\ahat)$, where
$\ahat=\ahat(\bu)$
solves the dual problem
\begin{equation}
 \min_{\bsalpha \in \bbR^n}~\left\{\Vint{\etad\left(
\bsalpha^{T}\bm \right)} - \bsalpha^T \bu\right\} 
\:. \label{eq:dual_optimization}
\end{equation}
It is also the Legendre dual variable of $\bu$ with respect to the
strictly
convex entropy $h(\bu): = \cH(\cB(\ahat(\bu)))$, i.e.,
\begin{equation}
\ahat(\bu) = \left [\frac{\p h}{\p \bu}(\bu) \right ]^T \,.
\end{equation}
The moment system derived by setting $\cE(\bu) =
\cB(\ahat)$ in
\eqref{eq:moment_eqns_closed} is hyperbolic and symmetric when expressed in
the $\ahat$ variables and its solution formally dissipates $h$. Moreover, $\cE$
is an inherently positive quantity.
\end{theorem}

\begin{proof}
The form of the minimizer in \eqref {eq:entropy_ansatz} can be derived formally
using standard Lagrange multiplier techniques.  However, a rigorous proof
requires more technical arguments which can be found, for example, in
\cite{Junk-2000} for the Maxwell-Boltzmann entropy and applied directly to the
current setting.  Once the existence of a minimizer is found, the other
properties can be verified, as is done in
\cite{Dubroca-Feugas-1999,Levermore-1996}
\end{proof}

\subsection{Perturbed Entropy-Based (PEB) Closures}
\label{sec:PEB}

Perturbations to standard \pn closures%
\footnote{These closures are based on a spherical harmonic expansion in angle
and can be formulated as an entropy-based closure with an $L^2$ cost
functional \cite{Hauck-McClarren-2010,Pomraning-1973}.
}
have been
derived for
$N=3$ in \cite{Oh-Holloway-2009} and for general $N$ in
\cite{SchFraLev11} (see also \cite{Hauck-2006} and \cite{Struchtrup-2008}
). 
The idea
behind the derivation in \cite{SchFraLev11} is to write $\psi
= \psi_{\rm{pn}} + \psit$,
where $\psi_{\rm{pn}}$ is the standard $P_N$ expansion.  The perturbation
$\psit$ satisfies its own kinetic equation, which can be then used to
approximate $\psit$ in terms of $\psi_{\rm{pn}}$. The resulting ``$D_N$''
models gain a diffusive term in the equations for the highest order moments.
 uch an approach need not be restricted to the $P_N$ equations.  Indeed,
following this exact strategy, we define
\begin{enumerate}
 \item The moment map $\cM:  g \mapsto \bu :=\vint{\bm g}$;
 \item The expansion map $\cE: \bu \mapsto \cB(\ahat(\bu))$;
 \item The reconstruction $\cR = \cE \circ \cM$;
 \item The kinetic perturbation $ \psit = \psi - \cR(\psi)$.
\end{enumerate}
The kinetic equation for $\psit$ is
\begin{equation} \label{eq:diff}
 \p_t \psit =  \p_t \psi - \p_t \cR(\psi)=  \p_t \psi - \p_t \cE(\bu)
  =  \p_t \psi -\cE'(\bu) \p_t\bu
\end{equation}
where
\begin{equation}
 \cE'(\bu)
  = \cB'(\ahat)
     \frac{\p \ahat}{\p \bu} 
  = \bm^T \cW(\bu)
    \Vint{\bm \bm^T \cW(\bu)}^{-1}
\,,
    \label{eq:E_prime}
\end{equation}
\begin{equation}
  \cW(\bu) := \etad''( \ahat^{T}\bm )
  =\frac{2h^2 \nu^4} {kc}\frac{\exp(-\frac{h\nu c}{k }
\ahat^{T}\mathbf{m})}
        {\left[\exp(- \frac{h\nu c}{k }\ahat^{T}\mathbf{m})-1\right]^2}
> 0 \:,
\end{equation}
and we have used the relation
\begin{equation}
{\rm Id}  = \Vint{\bm \cE'(\bu)} 
=  \Vint{\bm \cB'(\ahat) }  \frac{\p\ahat}{\p\bu}
=  \Vint{\bm \bm^T \cW(\bu) }  \frac{\p\ahat}{\p\bu}
\end{equation}
to compute the matrix $\frac{\p\ahat}{\p\bu}$ in \eqref{eq:E_prime}.
By operating with $\Pt := \cI - \P$ on \eqref{eq:transport} , where $\P:=\cE'(\bu)\cM$, we can write $\eqref{eq:diff}$ as
\begin{equation} \label{eq:deviation}
 \frac1c\p_t \psit + \Pt(\Omega \cdot \grad_x \psi) = \Pt \cC( \psi; T).
\end{equation}
It should be noted, for future use, that the
projection $\Q$, given by
\begin{equation}
\Q g:= \frac{1}{\cW(\bu)} \P(\cW(\bu) g)  \:,
\label{eq:Q}
\end{equation}
is self-adjoint in $L^2$ with respect to the
positive weight $\cW(\bu)$.

Equation \eqref{eq:deviation} for the perturbation is exact.  To derive a
closure, we neglect the
time derivative and perturbative component of the flux to arrive at the
following approximate balance equation
\begin{equation} \label{eq:balance}
 \Pt(\Omega \cdot \grad_x \E) \simeq \Pt \cC( \psi; T),
\end{equation}
where
\begin{equation}
 \Pt \cC( \psi; T) = -\sigt \left[ \Pt \cE(\bu) + \psit \right]
    + \frac{1}{4 \pi} \left[ \sigs \Pt \phi + \siga \Pt B(T) +\Pt s \right].
\end{equation}
In the appendix, we show that for the grey equations with Bose Einstein
entropy,
\begin{equation}
\int_0^{\infty} \Pt \cE(\bu) \,d \nu = 0.
\label{eq:p_id}
\end{equation}
Knowing that this component will be integrated out in the final closure, we
therefore solve \eqref{eq:balance} for $\Pt \cE(\bu) + \psit$ in terms of a
convective component $\psic$ and a diffusive component $\psid$:
\begin{equation}
 \Pt \cE(\bu) + \psit  \simeq
   \frac{1}{4 \pi} \left[ \rs \Pt \phi + \ra \Pt B(T) + \frac{1}{\sigt}
\Pt s \right ]
   -\frac{1}{\sigt}\Pt(\Omega \cdot \grad_x \E)
    =: \psic + \psid ,
  \label{eq:deviation_approx}  
\end{equation}
where $\rs$ and $\ra$ are the scattering and absorption ratios, respectively:
\begin{equation}
\rs= \frac{\sigs}{\sigt} \quand \ra= \frac{\siga}{\sigt} \:.
\end{equation}
Inserting \eqref{eq:deviation_approx} back into the flux term of
the moment equation \eqref{eq:moment_eqns_closed} gives
\begin{equation}
 \vint{\Omega \bm \psi} \simeq \vint{\Omega \bm \E} + \vint{\Omega \bm \psic} +
\vint{\Omega \bm \psid} =: \fe + \fc + \fd.
\end{equation}

%

At this point, it is not clear whether this flux dissipates an entropy or if the
convective flux $\fc$ is always hyperbolic. In general, the hyperbolicity of
moment models closed by an entropy minimization principle follows from the fact
that (in terms the Lagrange multipliers $\ahat$) the model can be written as a
symmetric Lax-Friedrichs form \cite{Levermore-1996}. This structure is not
present here.  However, at least for slab geometries, the convective flux
in the $PM_1$ model is hyperbolic.  (See Proposition \ref{prop:pM1_hyperbolic}
in the following section.)  Moreover, in general, the diffusive flux
satisfies a local dissipation law.   
\begin{prop}
\label{prop:dissipation}
 The diffusion term $\fd$ dissipates the entropy $h(\bu):= \cH(\E)$
locally in space.
\end{prop}

\begin{proof}
A dissipation law for $h$ is found by multiplying
the closed moment system \eqref{eq:moment_eqns_closed} by
$ \ahat^T \equiv \frac{\p h}{\p \bu}$.  Multiplying $\grad_x\cdot\fd$ on the right by
$\ahat^T$ gives
 \begin{align*}
 \ahat^T \left ( \grad_x\cdot\fd \right )&=
 -\ahat ^T \left [ \div \Vint{\Omega \bm \sigt^{-1} \Pt\left(\Omega \cdot \grad_x
\E\right)} \right ] \\
 &= -\div \Vint{\Omega (\ahat ^T \bm) \sigt^{-1}  \Pt\left(\Omega \cdot \grad_x
\E\right)}
    + \left(\grad_x \ahat^T \right) \cdot \Vint{\Omega \bm \sigt^{-1}  \Pt\left(\Omega \cdot
\grad_x
\E\right)},
 \end{align*}
where $\grad_x$ acts on the components of $\Omega$ and the Lagrange multiplier $\ahat^T$ on $\bm$.
We only need to work with the term that is not in divergence form.  We use the
fact that $\cB(\ahat) = (h \nu c /k) \bm^T \cW$
\begin{align*}
 \left(\grad_x \ahat^T \right) \cdot \Vint{\Omega \bm \sigt^{-1} \Pt\left(\Omega \cdot \grad_x
\E\right)}
  &= \left(\grad_x \ahat^T \right) \cdot
      \Vint{\Omega \bm \sigt^{-1} \cW \Qt
        \left( \frac{\Omega \cdot \grad_x \E  }{\cW} \right)} \nonumber \\
  &= \frac{h \nu c}{k}\Vint{\Omega \cdot \grad_x ( \ahat^T \bm) \sigt^{-1}  \cW
       \Qt \left( \Omega \cdot \grad_x (\ahat^T \bm) \right)} \nonumber\\
  &=  
      \frac{h \nu c}{k}\Vint{\sigt^{-1} \cW \left[\Qt( \div (\Omega \ahat^T
\bm)\right]^2}\nonumber \\
      & \geq 0\:,
 \end{align*}
where $\Qt := \Id - \Q$ and $\Q$ is given in \eqref{eq:Q}.
\end{proof}

%

\subsection{Controlling the Perturbations}
\label{subsec:pert}

While the entropy-based ansatz in \eqref{eq:entropy_ansatz} is positive for
all $\Omega$, the addition of the perturbation in \eqref{eq:deviation_approx}
may lead to an ansatz which is not.   As a consequence, the moments
of the perturbed ansatz may not satisfy the realizability condition
\eqref{eq:real}.  
To correct for this defect, we introduce a modification and approximate $\psi$
with
\begin{equation}
  \E = \cB(\ahat) + \delta \psit,
  \label{eq:cut_ansatz}
\end{equation}
where $\delta(x,t)$ is  a scalar control parameter.  Several
different choices for $\delta$ are possible.  For example, one could select
it to ensure that $\E$ is positive everywhere.  However, this choice
requires 
pointwise evaluations with respect to $\Omega$---a task we would like to avoid. 
Instead, we select $\delta$ in such a way as to preserve \eqref{eq:real} in the numerical
computation.  While the exact form of $\delta$ depends on the details of the
numerical method, the general framework relies on the realizability conditions
for the moments. We call an array $(\Psi_0,\Psi_1, \ldots , \Psi_N)$ realizable with
respect to $(1,\Omega, \ldots,\Omega^{\otimes N})$ if there exists a
non-negative measure  on $d\Omega d\nu$ with density $\Psi(\Omega,\nu)$ such
that $\Psi_k = \vint{\Omega^{\otimes k}  \Psi}$ for $k=1, \ldots, N$. The set
$\cR_N$ of all such vectors is called the realizable set.

Roughly speaking, we select $\delta$ to ensure that $(\Psi_0,\Psi_1, \ldots ,
\Psi_N) \in \cR_N$.  Note that such a $\delta$
always exists: When $\delta = 0$, there is no perturbative term and since the minimum entropy
ansatz is always positive, $(\Psi_0,\Psi_1, \ldots , \Psi_N) \in \cR_N$. 
Details for the $PM_1$ are given in Section \ref{subsec:limiters}.


\section{The Perturbed $M_1$ ($P\!M_1$) model}
The perturbed $M_1$ model is based on the moments
\begin{equation}
\bu =  \left(
\begin{array}{c}
\bu_0 \\
\bu_1
\end{array}
 \right)
= \left(
\begin{array}{c}
cE \\
F
\end{array}
 \right)
 := \left(
\begin{array}{c}
\vint{\psi} \\
\vint{\Omega \psi}
\end{array}
 \right) \:.
\end{equation}
where $E$ is the photon energy density and $F$ is the energy flux density.
The model approximates the evolution of $E$ and $F$ with the following system:
\begin{subequations}
\begin{align}
\p_t E + \div F &= -\siga (cE - acT^4) + S, \\
\p_t F   + c^2\div \Pi(E,F) &= -c\sigt F \,,
\end{align}
\end{subequations}
where $S(x,t):= \int_0^\infty s(x,\nu,t) \, d\nu$ and the closure for the
pressure term is
\begin{equation}
\Pi(E,F) := \frac{1}{c}\vint{(\Omega \vee \Omega) (\cE(\bu) + \psic + \psid)}
=: \pe(E,F) + \pc(E,F) + \pd(E,F) \:.
\end{equation}
Here $\pe(\bu)$ is the term that comes from the entropy ansatz (the
\textit{entropy-based term}).  The term $\pc(\bu)$ is the \textit{convective
correction} and $\pd(\bu)$ is the \textit{diffusive correction}.  These
corrections can be expressed in terms of $\pe$ and
\begin{equation}
\qe:= \Vint{\Omega^{\vee 3}\E}
\end{equation}
which, in turn, can be expressed in terms of the
unit vector $\bn:= F / |F|$ and the scalars
\begin{equation}
\chi_k = \frac{\vint{(\Omega \cdot \bn)^k \cE}}{cE} \:.
\label{eq:eddingtons}
\end{equation}

\begin{lem} \label{lem:pd_pc}
The correction terms $\pc$ and $\pd$ are given by
\begin{subequations}
\begin{align}
\pd &=  \frac{1}{c \sigt}\left[-\div \qe
 + \frac{ \p \pe}{\p E}  (\div F)
 + c^2 \frac{ \p \pe}{\p F}  (\div \pe) \right]
 \:, \\
\pc &= \eta \cdot \left( \rs E + \ra a T^4  + \frac{S}{c \sigt} \right), \q \tx{where} \q \eta=\left( \third \Id -  \frac{ \p \pe}{\p E} \right).  \label{eq:eta_abstract}
\end{align}
\end{subequations}
\end{lem}
\begin{proof}
see appendix.
\end{proof}

\begin{rem}The formula for the convective
correction is independent of the specific form of $\cE$.  In particular for the
$P_1$ model, the pressure term is $\pp = \third E$ so that $\frac{ \p
\pp}{\p E} = \third \Id$ and $\pc = 0$; cf. \eqref{eq:eta_abstract}. This is
consistent with the fact that the ``$D_N$'' models in
\cite{SchFraLev11} contain only diffusive corrections.
\end{rem}

\begin{lem} \label{lem:pe_qe}
The entropy-based terms $\pe$ and $\qe$ are given by
\begin{subequations}
\begin{align}
 \pe &= \frac{E}{2}[(1- \chi_2) \Id + (3\chi_2-1)(\bn \vee \bn)] \:,\\
 \qe &= \frac{3cE}{2}[(\chi_1 - \chi_3) (\Id \vee \bn) + (5 \chi_3 - 3\chi_1)
\bn^{\vee 3}] \:,
\end{align}
\end{subequations}
where the scalars $\chi_1$, $\chi_2$, and
$\chi_3$ are defined in \eqref{eq:eddingtons}.
\end{lem}
\begin{proof}
See the appendix.
\end{proof}

In slab geometry, we end up with the following expressions for the components of the pressure term $\Pi = \pe + \pc + \pd$ which can be computed from Lemma~\ref{lem:pd_pc} and Lemma~\ref{lem:pe_qe}:
\begin{gather}
\label{eq:pressures}
\pe = \chi(E,F) E\:, \q
\pc = \ds  \eta(E,F)
\left( \rs E + \ra a T^4  + \frac{S}{c \sigt} \right) \:,\\
 \pd = -\frac{1}{c \sig{t}} [ D_E(E,F) \p_x E
+D_F(E,F) \p_x F ] =: {\bf D} (\bu) \p_x \bu \:,
\end{gather}
where the convection and
diffusion coefficients are given by
\begin{gather}
\chi(E,F) = \frac{1+3\gamma^2}{3+\gamma^2}, \q\q \eta = \frac{8\gamma
^2}{3(3-\gamma ^2)}, \label{eq:chi_eta} \\
D_E(E,F) = \frac{ 3 (\gamma ^2+5) \, (\gamma ^2-1)^2 }{2\gamma ^4 (\gamma
 ^2-3)^2}  \left [ (\gamma ^2-3)\ln \left(\ds \frac{1-\gamma}{1+\gamma}
 \right)-6\gamma \right], \\
 D_F (E,F) =  \frac{9(\gamma^2+1) (\gamma^2-1)^2}{2\gamma ^5(\gamma ^2 -3)^2}
 \left [ (\gamma ^2-3) \ln \left( \frac{1-\gamma}{1+\gamma} \right) -6\gamma
 \right ],
\end{gather}
and
\begin{equation}
     \gamma = \frac{-3F}{2cE+\sqrt{(2cE)^2-3F^2}}.
\end{equation}

These coefficients are displayed in \reffig{pict:ModelCoef}. Note that 
$\chi$, $\eta$, and $D_F$ are all even functions of the ratio $F/(cE)$, while
$D_E$ is odd.

\begin{figure}
\centering
\subfloat[Convection coefficients.]{\includegraphics[scale=0.42]{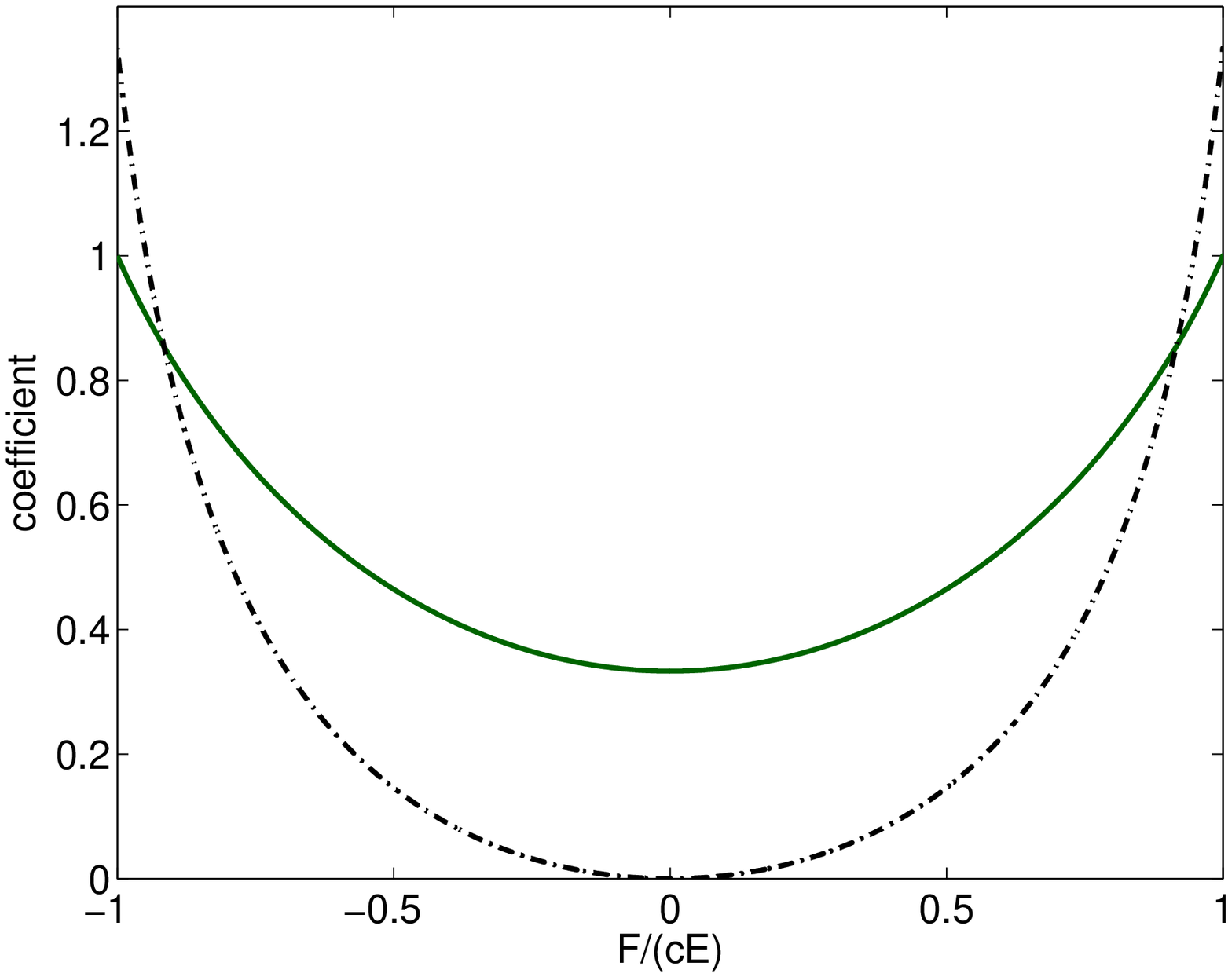}}
\,
\subfloat[Diffusion coefficients.]{\includegraphics[scale=0.42]{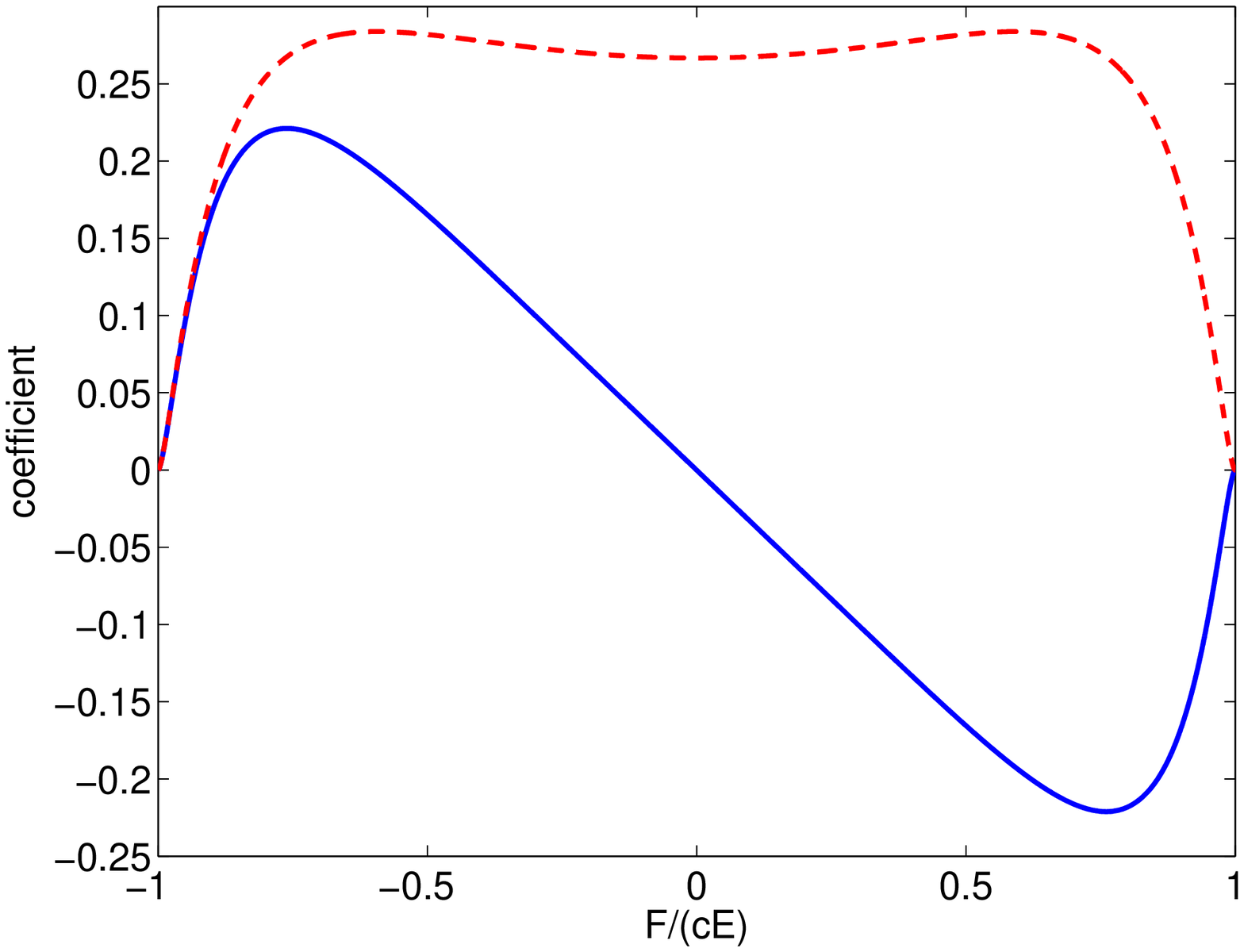}}
\caption{Perturbed M1 model coefficients. Left: $\chi$ (dark green solid line)
and  $\eta$ (black dash-dot line).  Right:  $D_E$ (blue solid line) and $D_F$
(red dashed line).}
\label{pict:ModelCoef}
\end{figure}%

\section{Numerical Simulation using Discontinuous Galerkin}
\label{sec:dg}

In slab geometries, the diffusion-corrected $M_1$ model and the material energy
\eqref{eq:material_energy} reduce to
\seqalign{\label{eq:m1_system}}{
\p_t \bu + \p_x \bff(\bu, \p_x\bu) &= \bs(\bu), \quad (x,t) \in
(x_L,x_R)\times
(0,t_{\tx{final}}), \\
\p_t T &= \frac{c \sig{a}}{C_v} (E -aT^4),
\intertext{where}
\bu =\begin{bmatrix} cE \\ F \end{bmatrix},
\quad
\bs(\bu) &= \begin{bmatrix}
-c^2\sig{a}(E -aT^4) + cS \\
-c\sig{t} F
\end{bmatrix},
\q \bff(\bu, \p_x\bu) = \begin{bmatrix}
cF \\
c^2 \pdel
\end{bmatrix},
}
$\pdel = \pe + \delta (\pc + \pd)$ and $C_v = \frac{\p e}{\p T}$ is the specific
heat at constant volume.  In this setting, the convective flux of the $PM_1$
model is hyperbolic. 
\begin{prop}
\label{prop:pM1_hyperbolic}
The perturbed $M_1$ system in slab geometry is hyperbolic if $\pd=0$ and $|F| <
cE$.
\end{prop}
\begin{proof}
 The proof is a direct calculation, given in the appendix.
\end{proof}

We simulate the system \eqref{eq:m1_system} using a Runge-Kutta discontinuous
Galerkin (RKDG) method.  The RKDG method is a method of lines: the DG
discretization is only applied to spatial variables while time discretization is
achieved by explicit Runge-Kutta time integrators. The presentation here is
rather brief and relies on details found in \cite{OlbHauFra11}, where the method
was applied to the $M_1$ model. A general description of the RKDG method can be
found, for example, in \cite{CocLinShu89, Cockburn-Shu-Karniadakis-2000}.

\subsection{Spatial Discretization}

We divide the computational domain $(x_L,x_R)$ into $J$ cells with edges
\[ x_L=x_{1/2} < x_{3/2} < \ldots < x_{J+1/2}=x_R, \] 
and let $x_j$ denote the center of each cell $I_j=(x_{j-1/2}, x_{j+1/2})$.
We let $h_j:=x_{j+1/2}-x_{j-1/2}$ be the length of the interval $I_j$ and set
$h:=\max_j h_j$.  We denote the finite-dimensional approximation space by
\[ V_h^k = \{v\in L^1(x_L,x_R): \, v_{|_{I_j}} \in {\cal P}^k(I_j), \,
j=1,\dots,J \}, \]
where ${\cal P}^k(I_j)$ is the space of polynomials of degree at most $k$ on the
interval $I_j$.

The semidiscrete DG scheme is derived from a weak formulation of
\eqref{eq:m1_system}.
However, following \cite{CocShu98} we first reduce the convection-diffusion
equations
\eqref{eq:m1_system} to a system of first-order equations by introducing the
auxiliary variable $\bv$:
\seqalign{}{
\p_t \bu + \p_x \bff(\bu, \bv) &= \bs(\bu), \label{eq:hyperbolic} \\
 \p_x \bu &= \bv, \\
 \p_t T &= \frac{c \sig{a}}{C_v} (E -aT^4).
}
The exact solutions $\bu(\cdot,t)$, $\bv (\cdot,t)$ and $T(\cdot,t)$
are then replaced by approximations $\bu_h(\cdot,t)$, $\bv_h(\cdot,t)$ $\in
V_h^k\times V_h^k$ and $T_h(\cdot,t)\in V_h^k$, and the resulting set of
equations is required to hold for all test functions
$b_h\in V_h^k$:
\seqalign{\label{eq:weakform}}{
\int_{I_j} b_h(x)\p_t \bu_h(x,t) dx - \int_{I_j}
\bff(\bu_h(x,t),\bv_h(x,t))
\p_x
b_h(x) dx & \\
+ \left \llbracket \bff b_h(x) \right\rrbracket_j &=  \int_{I_j}
\bs(\bu_h(x,t))b_h(x)  dx \nonumber \\
\int_{I_j} b_h(x) \bv_h(x,t) dx + \int_{I_j} \bu_h(x,t) \p_x b_h(x)
dx - \left \llbracket \bu b_h(x) \right\rrbracket_j &=  0 \\
\int_{I_j} b_h(x) \p_t T_h(x,t) dx =\int_{I_j} b_h(x) \frac{c
\sig{a}(x)}{C_v} (E_h(x,t) &-aT_h^4(x,t)) dx.
}
Here we use the bracket notation:
\begin{align}
\left \llbracket \bff b_h(x) \right \rrbracket_j=
\bff_{j+1/2}b_h(x_{j+1/2}^-)
-\bff_{j-1/2}b_h(x_{j-1/2}^+)
\end{align}
where
\begin{equation}
\label{eq:edge_fluxes}
\bff_{j\pm 1/2}(\bu, \bv) = \bff(\bu(x_{j\pm 1/2},t), \bv(x_{j\pm
1/2},t)) 
\end{equation}
and
\begin{align}
\label{eq:rl_limits}
b_h(x_{j + 1/2}^-) = \lim_{\eps \to 0^+} b_h
(x_{j+1/2} -\eps), \quad b_h(x_{j - 1/2}^+) = \lim_{\eps \to 0^+}
b_h (x_{j-1/2} +\eps)
\end{align}
are the right and left limits of $b_h$ at the cell interfaces $x_{j\pm
1/2}$.  The term
$\left \llbracket \bu b_h(x) \right\rrbracket_j$ is defined in an
analogous fashion.

Since the components of $\bu_h(.,t)$ and $\bv_h(.,t)$ are piecewise polynomials,
the edge values
of $\bu$ and $\bv$ in \eqref{eq:edge_fluxes} are not strictly defined. 
Thus, the nonlinear flux function $\bff$ is replaced
by a numerical flux
$\hat{\bff}$
which depends on  the pointwise limits of $\bu_h$, $\bv_h$ on either side of the
edge at
$x_{j\pm1/2}$: 
\begin{align}
\hat{\bff}_{j \pm1/2} = \hat{\bff}(\bu_h(x^-_{j\pm1/2},t),
\bu_h(x^+_{j\pm1/2},t),
\bv_h(x^-_{j\pm1/2},t), \bv_h(x^+_{j\pm1/2},t)).
\end{align}
The notations for $\hat{\bu}$ carry over analogously.

\begin{figure}
\centering
\includegraphics[scale=0.45]{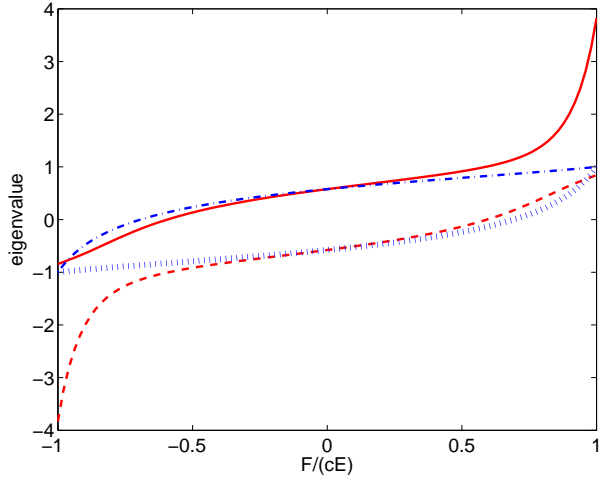}
\caption{Eigenvalues of the hyperbolic flux Jacobian: M1 model (blue lines),
perturbed M1 model (red lines).} \label{pict:EigVal}
\end{figure}%

It remains to choose suitable numerical fluxes $\hat{\bff}$ and $\hat{\bu}$.
Since \eqref{eq:m1_system} has both a convective flux
\begin{equation} 
\fc(\bu) := \begin{bmatrix}
cF \\
c^2 \, \pe + \ds \frac{c}{\sig{t}} \eta(E,F)  \left ( c\sig{s}  E +
c\sig{a} a T^4 + S \right )
\end{bmatrix}
\label{eq:conv_flux}
\end{equation}
and a diffusive flux
\begin{equation} 
\fd(\bu, \bv) = \begin{bmatrix}
0 \\
c^2 \, {\bf D} (\bu) \cdot \bv
\end{bmatrix} \:
\end{equation}
the choice is not obvious. Several approaches have been presented in
literature \cite{BasReb97,XuShu10,LiuYan09,NguPerCoc09}.
In \cite{BasReb97}, the prescription for the diffusive term is given by
\begin{align}
\fdh_{j\pm 1/2} &= \frac12 \left [ \fd(\bu^-_{j\pm 1/2} ,
\bv^-_{j\pm 1/2} ) + \fd(\bu^+_{j\pm 1/2} , \bv^+_{j\pm 1/2} ) \right ], \\
\hat{\bu}_{j\pm 1/2} &= \frac{1}{2} \left [ \bu_{j\pm 1/2}^- + \bu_{j\pm 1/2}^+
\right ].
\end{align}
Combining this term with the Lax-Friedrichs flux
for $\fc(\bu)$ gives the following total numerical flux:
\begin{align}
\label{eq:LFflux}
\hat{\bff}_{j\pm 1/2} = \frac{1}{2} \left [ \bff(\bu_{j\pm 1/2}^-,\bv_{j\pm
1/2}^-) +\bff(\bu_{j\pm 1/2}^+,\bv_{j\pm 1/2}^+)
-\lambda
(\bu_{j\pm 1/2}^+ -\bu_{j\pm 1/2}^-) \right ] \, ,
\end{align}
where $\lambda$ is the largest
magnitude of any eigenvalue of the Jacobian associated with $\fc$.  
These eigenvalues, in general, depend on material properties, the temperature
$T$ and the source term $S$.
In contrast to the M1 model, they are not bounded by the speed of light $c$.
For example, neglecting the temperature and source the maximum value is
approximately $9.12\, c$. We instead use the smaller value of $\lambda
= c$, which is the particle speed in the transport equation and is consistent
with the application of the control parameter to enforce realizability (see
Section \ref{subsec:pert}). \reffig{pict:EigVal} shows the comparison of
eigenvalues for the M1
and perturbed M1 modeling when $c=1$, $\sig{s}=1$, $\sig{t}=3$, $T=0=S$.

%

The DG solutions $\bu_h$, $\bv_h$ and $T_h$ are expanded in terms of local basis
functions
$\{
b^j_l \}_{l=0}^k$ for ${\cal P}^k(I_j)$ in each cell $I_j$:
\begin{align*}
\bu^j_h(x,t) = \sum_{l=0}^k \bu^j_l(t) b^j_l (x), \q \bv^j_h(x,t) =
\sum_{l=0}^k \bv^j_l(t) b^j_l (x), \q T^j_h(x,t) = \sum_{l=0}^k T^j_l(t) b^j_l
(x) \quad \text{for } x\in I_j.
\end{align*}
%
%
The standard choice of basis for ${\cal P}(I_j)$ is generated by
Legendre polynomials $P_l$ that are defined on the reference cell $[-1,1]$:
\begin{align} 
b^j_{l}(x) = P_l\left ( \frac{2(x-x_j)}{h_j} \right )\:,
\quad  j \in \{1,\dots,J\}\,, l\in\{0,\dots,k\}\:,
\end{align}
and normalized so that
\begin{equation}
 \int_{-1}^1 P_l(y) P_m(y) dy = \frac{2}{2m+1} \delta_{l,m}.
\end{equation}
With  $\xi_j(y) := x_j + y h_j/2$, this gives a formulation
defined on the reference cell:
\seqalign{\label{eq:spaceDG}}{
&\frac{h_j}{2m+1} \p_t \bu^j_{m}(t) - \int_{-1}^1 \bff(\bu^j_h(\xi_j(y),t),
\bv^j_h(\xi_j(y),t)) \partial_y P_m(y) dy   \label{eq:spaceDG_u} \\ 
& \hspace{3cm} + \hat{\bff}_{j+1/2} - (-1)^m \hat{\bff}_{j-1/2} =
\frac{h_j}{2}\int_{-1}^1 \bs(\bu^j_h(\xi_j(y),t)) P_m(y) dy, \nonumber \\
&\frac{h_j}{2m+1} \bv^j_{m}(t) + \sum_{l=0}^k \bu_l^j(t) \: {\cal C}_{l,m} -
\hat{\bu}_{j+1/2} + (-1)^m \hat{\bu}_{j-1/2} = 0,  \label{eq:spaceDG_w} \\
& \frac{h_j}{2m+1} \p_t T_m^j(t)  = \frac{c h_j}{2 C_v} \sum_{l=0}^k E_l^j(t)
\left ( \int_{-1}^1 P_m(y) P_l(y) \sig{a}(\xi_j(y)) \right ) dy
\label{eq:spaceDG_T} \\
 &\hspace{9cm} -\frac{ach_j}{2C_v} \int_{-1}^1 P_m(y) {T_h^j}^4(\xi_j(y),t) \sig{a}(y) dy, \nonumber
}
where
\balign{
{\cal C}_{l,m} = \int_{-1}^1 P_l(y)  \p_y P_m(y) dy.
}
The remaining integrals are calculated by a quadrature rule.  Note that
\eqref{eq:spaceDG_w} can be solved locally for $\bv^j_{m}(t)$ in each
cell $I_j$, which can then be substituted back into \eqref{eq:spaceDG_u}.

\subsection{Time Discretization: Explicit SSP Runge-Kutta Schemes}
\label{subsect:RK}
The purpose of high-order, strong stability preserving (SSP) Runge-Kutta time
integration methods is to achieve high-order accuracy in time while preserving
desirable properties of the forward Euler method (for a review, see
\cite{GotKetShu09}). In this paper, we only use
explicit schemes, which compute values of the unknowns at several
intermediate stages. Each stage is a convex combination of forward Euler
operators and this usually leads to modified CFL restrictions.
%
%

The equations in \eqref{eq:spaceDG} form a system of ODEs
for the coefficients $\bu^j_{m}(t)$ and $T_m^j(t)$.  
For all $j \in \{ 1,\ldots,J\}$ and $m \in \{0,\ldots,k\}$, we write this system
in the abstract form:
\begin{align}
\p_t \bu^j_{m}(t) &= {\cal L}^j_{\bu, m} (\bu^{j-1}_0,\ldots,\bu^{j-1}_k,
\bu^j_0,\ldots,\bu^j_k,\bu^{j+1}_0,\ldots,\bu^{j+1}_k ), \\
\p_t T^j_{m}(t) &= {\cal L}^j_{T,m} (E^j_0,\ldots,E^j_k, T^j_0,\ldots,T^j_k ).
\end{align}
Here, ${\cal L}^j_{\bu,m}$ and ${\cal L}^j_{T,m}$ are the respective right-hand
sides of the ODEs.

Let $\{t^n\}_{n=0}^N$ be an equidistant partition of
$[0,t_{\tx{final}}]$ and set $\Delta t := t_{\tx{final}}/N$.
Let $\Lambda$ denote the application of a generic slope limiter, and 
let $\pi_{V^m_h}$ be the orthogonal projection onto the finite dimensional space
$V_h^m$.
Note that $\Lambda$ is applied at every Runge-Kutta stage.
The algorithm for the optimal third-order SSP Runge-Kutta (SSPRK(3,3))
method \cite{ShuOsh89} reads as follows: 
\begin{itemize}
\item For all $j \in \{1,\ldots,J\}$ and $m\in \{0,\ldots,k\}$, set 
$\bu_m^{j,0} = \Lambda \{ \pi_{V^m_h}(\bu_0) \}$.
\item For all $ n \in \{0,\ldots,N-1\}$, $j \in \{1,\ldots,J\}$, and $m\in
\{0,\ldots,k\}$,
\begin{enumerate}
\item Compute the intermediate stages
\begin{align}
\bu_m^{j,(1)} &= \Lambda \left\{ \bu_m^{j,n}  +\dt  {\cal
L}^j_{\bu,m}(\bu_h^{j,(0)}) \right \} \nonumber	\\
\bu_m^{j,(2)} &= \Lambda \left\{ \frac{3}{4} \bu_m^{j,n}  +\frac{1}{4} 
\bu_m^{j,(1)} +\frac{1}{4} \dt  {\cal L}^j_{\bu,m}(\bu_h^{j,(1)}) \right\}  
\label{SSP3} \\
\bu_m^{j,(3)} &= \Lambda \left \{ \frac{1}{3} \bu_m^{j,n}  +\frac{2}{3} 
\bu_m^{j,(2)} +\frac{2}{3} \dt  {\cal L}^j_{\bu,m}(\bu_h^{j,(2)}) \right\} 
\nonumber
\end{align}
\item Set $\bu_m^{j,n+1} = \bu_m^{j,(3)}$.
\end{enumerate}
\end{itemize}

\noindent For the sake of completeness, we also state the optimal fourth order scheme
SSPRK(5,4) \cite{GotKetShu09}:
\begin{align}
\bu_m^{j,(1)} &= \Lambda  \{ \bu_h^{(n)}  + 0.391752226571890 \, \dt  \, {\cal
L}^j_{\bu,m}(\bu_m^{j,(n)}  )  \}   \nonumber\\
\bu_m^{j,(2)}  &= \Lambda \{ 0.444370493651235 \, \bu_m^{j,(n)}  +
0.555629506348765 \, \bu_m^{j,(1)}   \nonumber \\
& \quad + 0.368410593050371 \, \dt \,  {\cal L}^j_{\bu,m}(\bu_m^{j,(1)}  )
\} \nonumber \\
\bu_m^{j,(3)}  &= \Lambda  \{ 0.620101851488403 \, \bu_m^{j,(n)}  +
0.379898148511597 \, \bu_m^{j,(2)}   \nonumber \\
& \quad + 0.251891774271694  \, \dt \,  {\cal L}^j_{\bu,m}(\bu_m^{j,(2)}  )
\}  \label{SSP4} \\
\bu_m^{j,(4)}  &= \Lambda \{ 0.178079954393132 \, \bu_m^{j,(n)}  +
0.821920045606868 \, \bu_m^{j,(3)} \nonumber \\
& \quad + 0.544974750228521 \, \dt  \, {\cal L}^j_{\bu,m}(\bu_m^{j,(3)}  ) \}
\nonumber  \\
\bu_h^{(n+1)}  &= \Lambda  \{ 0.517231671970585 \, \bu_m^{j,(2)}
+0.096059710526147 \, \bu_m^{j,(3)}  \nonumber \\ 
& \quad +0.386708617503269 \, \bu_m^{j,(4)} +0.063692468666290 \, \dt \,
{\cal L}^j_{\bu,m}(\bu_m^{j,(3)} )  \nonumber \\ 
& \quad +0.226007483236906 \, \dt \, {\cal L}^j_{\bu,m}(\bu_m^{j,(4)} )  \}.
\nonumber 
\end{align}
Note that SSPRK(3,3) permits a timestep of the same size as forward
Euler, while the SSPRK(5,4) method is less restrictive, allowing for a time step
that is $1.508$ times larger the forward Euler scheme.



\subsection{Limiters} \label{subsec:limiters}
As in \cite{OlbHauFra11}, two types of limiters are used.  The first is
standard; it is used to suppress spurious oscillations and maintain stability. 
There are many such limiters available. In this paper, we apply the moment
limiter from \cite{BurSagBru01}, which is a modification of the original limiter
in \cite{Biswasa-Devine-Flaherty-94}. This limiter is applied to the variables
$\bu$, but not the auxiliary variables $\bv$ or the temperature $T$. Additional
details can be found in \cite{OlbHauFra11}.

\subsubsection{Realizability-Preserving Limiter}
The second limiter is a \textit{realizability-preserving limiter} which is
needed to ensure that the cell averages of $E$ and $F$ satisfy the realizability
condition \eqref{eq:real} at each stage of the numerical computation.  The
limiter is based on the work from \cite{ZhaShu10a} and
\cite{ZhaShu10b} and is very similar to what was done in \cite{OlbHauFra11} for
the $M_1$ model. The major difference here is the addition of the control
parameter $\delta$.

An essential ingredient of the realizability limiter is the Gauss-Lobatto
quadrature set
\begin{equation} \{ x_{j-1/2}= \hat{x}_j^1, \hat{x}_j^2,\ldots,
\hat{x}_j^{M-1}, \hat{x}_j^M = x_{j+1/2} \} \subset I_j, \end{equation}
where, for a spatial reconstruction of order $k$, $M$ is the smallest integer
such that $2M-3 \geq 2k + 1$.  This condition on $M$ ensures accuracy of the
scheme \cite{CocHouShu89}.  The weaker condition $2M-3 \geq k$ ensures that the
quadrature integrates elements of the approximation space $V_h^k$ exactly.

The realizability limiter is defined in order to ensure that
$\bu_h(\hat{x}_j^\ell) \in \cR_2$ at each point $\hat{x}_j^\ell$ in the
quadrature set. However, we enforce the convexity condition indirectly by
requiring the positivity of the intermediate quantities%
\footnote{The meaning of all subsequent subscripts, superscripts and
adornments of $Q$ and $R$ will be inherited from analogous definitions for $E$
and $F$.}%
\begin{equation}
 Q := \frac{c E +  F }{2}
\quand
R := \frac{c E - F }{2} \:.
\label{eq:QR}
 \end{equation}
The inverse transformation that maps $(Q,R) \mapsto (cE,F)$ is given by

\begin{equation}
E = \frac{ Q + R}{c}
\quand
F =  Q - R \:.
\label{eq:QRinv}
\end{equation}
An additional limiter is also used to enforce the positivity of the temperature
reconstruction at each quadrature point.

We now proceed to define the limiters. Let $\bu_h^{j,n} = (c E_h^{j,n},
F_h^{j,n})$ and $T_h^{j,n}$ be the approximations of $\bu$ and $T$ in cell
$I_j$ at time $t^n$, and let
$\hat{\bu}_h^{j,n}$ and $\hat{T}_h^{j,n}$ denote the modifications of
$\bu_h^{j,n}$ and $T_h^{j,n}$ that are generated by the
limiting.  We assume that the cell average of
$\bu_h^{j,n}$, which we denote
by $\wbar{\bu}_{h}^{j,n}$, is realizable, i.e., $\wbar{\bu}_{h}^{j,n} \in {\cal
R}_2$. We also assume that the cell average of
$T_h^{j,n}$, which we denote by $\wbar{T}_h^{j,n}$, is positive. Let
$Q_h^{j,n}(x)$
and $R_h^{j,n} (x)$ be the approximations of $Q$ and $R$,
respectively, and define limited variables by
\seqalign{\label{eq:intermQuant}}{
 \hat{Q}_h^{j,n}(x) &= \theta^{j,n}_Q Q_h^{j,n}(x) + (1-\theta^{j,n}_Q)
\wbar{Q}_h^{j,n} \:, \\
  \hat{R}_h^{j,n}(x) &= \theta^{j,n}_R R_h^{j,n}(x) + (1-\theta^{j,n}_R)
\wbar{R}_h^{j,n} \:, \\
	\hat{T}_h^{j,n}(x) &= \theta^{j,n}_T T_h^{j,n}(x) + (1-\theta^{j,n}_T)
\wbar{T}_h^{j,n} \:,
\intertext{where}
 \theta^{j,n}_Q &:= \min \left\{ \frac{\wbar{Q}_h^{j,n}-\veps
/2}{\wbar{Q}_h^{j,n}-Q^{j,n}_{\min}},1 \right\}
  \,,  \label{eq:thetaQ} \quad
 Q^{j,n}_{\min} := \min_{\ell=1,\ldots,M} Q^{j,n}_h(\xjalpha)\:, \\
 \theta^{j,n}_R&:= \min \left\{ \frac{\wbar{R}_h^{j,n}-\veps
/2}{\wbar{R}_h^{j,n}-R^{j,n}_{\min}},1 \right\}
  \,, \label{eq:thetaR} \quad
 R^{j,n}_{\min} := \min_{\ell=1,\ldots,M} R^{j,n}_h(\xjalpha)\:, \\
 \theta^{j,n}_T&:= \min \left\{
\frac{\wbar{T}_h^{j,n}}{\wbar{T}_h^{j,n}-T^{j,n}_{\min}},1 \right\}
  \,, \label{eq:thetaT} \quad
 T^{j,n}_{\min} := \min_{\ell=1,\ldots,M} T^{j,n}_h(\xjalpha)\:.
}
The parameter $\varepsilon >0$ is chosen to maintain numerical stability
with finite precision arithmetic; its value should be small relative to the
magnitude of the variables in a given problem. The components of
$\hat{\bu}_h^{j,n}$ are then defined using \eqref{eq:QRinv}. 
They satisfy the following property which is a key ingredient for maintaining
realizability in the RKDG scheme.

\begin{lem}[\cite{OlbHauFra11}]
If $\wbar{\bu}_h^{j,n} \in \cR_2$ (respectively: $\wbar{T}_h^{j,n} \geq 0$),
then $\hat{\bu}_h^{j,n}(\hat{x}_j^\ell) \in \cR_2^\veps:= \cR_2 + [\veps,0]^T$
(respectively: $\hat{T}_h^{j,n}(\hat{x}_j^\ell) \geq 0$) for $\ell =
1,\ldots,M$.
\end{lem}

\subsubsection{Setting the Control Parameter}
\label{sec_results_sub_sub:set_delta}

We now define the control parameter $\delta$, discussed in Section
\ref{subsec:pert}.  Our definition is guided by the following result.

\begin{lem}[\cite{OlbHauFra11}]
\label{lem:R3}
In the one dimensional setting, a necessary condition for
$(cE,F,c \Pi_\delta) \in \cR_3$ is that
\begin{itemize}
\item[(C1)] $\Pi_\delta \leq E$,
\item[(C2)] $|F \pm c \Pi_\delta | \leq cE \pm F$.
\end{itemize}
\end{lem}

Rather than to require $(cE,F,c \pdel) \in \cR_3$, we choose $\delta\in [0,1]$ to ensure
the weaker conditions (C1) and (C2). More specifically, for any $(cE,F) \in
\cR_2$, we set
\seqalign{\label{eq:delta}}{
\delta(E,F) &= \begin{cases}
				\delta_0(E,F),  \q &  \pc(E,F) + \pd(E,F) >0, \\
				\delta_1(E,F),  \q &  \pc(E,F) + \pd(E,F) <0, \\
              \end{cases}
\intertext{where}
\delta_0 = \min \left \{ \frac{E-\pe}{\pc + \pd}, 1 \right \}, \q \delta_1 &= \min \left \{ \frac{-2F+cE+c\pe}{c|\pc + \pd|}, \frac{2F+cE+c\pe}{c|\pc + \pd|}, 1 \right \}.
}

\begin{lem}
\label{lem:pressure}
For all $(cE,F) \in \cR_2$, $\pdel:= \pe
+ \delta[\pc + \pd]$ satisfies
(C1) -- (C2).
\end{lem}

\begin{proof}
The assertion $(cE,F)\in\cR_2$ implies $(cE,F, c\pe) \in \cR_3$.  It follows
then that for $\Pi^\tx{D}=0$, conditions (C1) -- (C2) are  trivially satisfied.
It remains only to show the following inequalities:
\begin{equation}
	c\pdel \leq cE
	\quand
	c\pdel \geq 2F-cE
	\quand
	c\pdel \geq -2F-cE.
\end{equation}
These relations are easily verified by applying the definition of
$\delta$ and using the fact that $(cE,F, c\pe) \in \cR_3$.
\end{proof}

With $\delta$ given by \eqref{eq:delta}, one can show that cell averages of
$\bu_h$ remain realizable and that the cell average of
$T _h$ remains positive in a forward Euler step. Let
\begin{equation}
	\bu^{j,n}_\ell := \bu_h^{j,n}(\xjalpha) \:,
\quad
T^{j,n}_\ell := T_h^{j,n}(\xjalpha)\:,
\quad
	\pdell^{j,n} := \pdel(\bu^{j,n}_\ell) \:,
\quad
	\sig{t,\ell} := \sig{t}(\xjalpha).
\end{equation}

\begin{lem}
\label{lem:CFL}
Assume that $2M -3 \geq k$ and for each $\ell=1,\ldots ,M$, that 
\begin{equation}
	\bu^{j,n}_\ell \in {\cal R}_2, \q T^{j,n}_\ell \geq 0
\end{equation}
and $\pdell^{j,n}$ satisfies (C1) and (C2).
Assume further that $\dt$ satisfies the following conditions:
 \begin{itemize}
\item[(A1)]{$\ds \dt < \min_{\ell=1,\ldots,M}{ \left \{ \frac{1}{c \sig{t,\ell}}
\right \} }$,}
\item[(A2)]{$\ds \dt <  \min_{\ell=1,\ldots,M}{ \left \{  \frac{w_\ell h}{ c(1
+w_\ell\sig{t,\ell}h  )}   \right \}, }  $ }
\item[(A3)]{$ \ds \dt \leq \min_{\ell=1,\ldots,M}{ \left \{ 
\frac{C_v}{\sig{a,\ell} ac(T^{j,n}_\ell)^3} \right \}. }$  }
 \end{itemize}
where $h:=\min_j h_j$.
Then after a forward Euler time step,
\begin{equation}
 \wbar{\bu}_{h}^{j,n+1} \in {\cal R}_2 \quand \wbar{T}^{j,n+1}_h \geq 0.
\end{equation}
\end{lem}

\begin{proof}
We refer the reader to the proof of Theorem 3 in \cite{OlbHauFra11} for the
$M_1$ model, which relies exactly on the conditions (A1)--(A3) and (C1)--(C2).
The only difference is that (C1) and (C2) are assumed in Lemma~\ref{lem:CFL}, while
in \cite{OlbHauFra11} they are naturally satisfied by the $M_1$ model.
\end{proof}

\begin{theorem}
The Runge-Kutta discontinuous Galerkin scheme which combines
\begin{enumerate}
	\item the space discretization in \eqref{eq:spaceDG},
	\item the limiters in \eqref{eq:intermQuant},
	\item the modified pressure $\pdel$ in \eqref{eq:m1_system} with
control parameter $\delta$ given by \eqref{eq:delta},
	\item a strong-stability-preserving Runge Kutta time integrator, and
	\item a sufficiently accurate Gauss-Lobatto quadrature
\end{enumerate}
preserves the realizability of the moments in the sense of cell averages. In
particular, if the time step
conditions (A1)-(A2) in the statement of Lemma \ref{lem:CFL} hold and
if $\wbar{\bu}^{j,n}_h \in \cR_2$, then $\wbar{\bu}^{j,n+1}_h \in \cR_2$.
\end{theorem}

\begin{proof}
Application of the limiters in \eqref{eq:intermQuant} ensures that the
conditions of Lemma \ref{lem:CFL} hold at each stage in the SSP-RK scheme.  
Each successive stage is an application of the forward Euler operator to the
current stage with an appropriately modified time step.  Thus, the conclusions
of Lemma \ref{lem:CFL} apply at the next stage, including the final stage,
which gives $\bu^{j,n+1}_h$.
\end{proof}

\section{Numerical Results}
\label{sec:results}
In this section, we present several numerical computations in slab geometry 
for a choice of test cases that are common
for the $M_1$ model. The goal is to compare and contrast the perturbed $M_1$
model with
the $M_1$ model and to point out benefits and drawbacks. Benchmark solutions are
generated by the discrete ordinates method with an upwind scheme in space, high-order spherical harmonics or
semi-analytic expressions. The RKDG implementation has been verified and
benchmarked in \cite{OlbHauFra11}. The correct implementation of the additional
perturbative terms has been checked by the method of manufactured solutions
\cite{SalKnu00}.


As in \cite{OlbHauFra11} our algorithm is implemented in MATLAB, and
Gauss-Lobatto quadrature on [-1,1] is used. Additionally, the Runge-Kutta time
integration methods as well as parameters for the admissibility limiter are
applied in the same way.
In order to satisfy the conditions of Lemma \ref{lem:CFL} and to guarantee
stability, the time step is set to
\seqalign{}{
& \dt  < \min \left \{ c_1, c_2, c_3, c_4  \right \}, \\
c_1 = \frac{1}{c\sig{t,\tx{max}}} \:, \quad
c_2 = &\frac{h \, w_{\tx{min}}}{c (1 +w_{\tx{max}} \sig{t,\tx{max}} \, h  )} \:, \q
c_3 = \frac{C_v}{ac \tau_{\tx{max}}} 
\: \quand
c_4 = \frac{h^2}{2(2k+1)},
}
where $k$ is the polynomial degree, $h=\min_{j}h_j$, $w_{\tx{min}}$ and $w_{\tx{max}}$ are the minimum and maximum quadrature weights, respectively. 
The quantities $\sigma_{t,\tx{max}}$ and $\tau_{\tx{max}}$ are the maximum
values of $\sig{t,\ell}$ and $\sig{a,\ell}(T_\ell^{j,n})^3$.

The constant $c_4$ is not needed to preserve realizability of the moments but
rather to enforce a parabolic CFL condition.  Without this condition, unstable
modes grow without bound until the control parameter $\delta$ turns on and
damps them.  Unfortunately, the parabolic CFL restriction leads to small time
steps.

The stability parameter for the realizability limiter of $E$ and $F$ is set to $\varepsilon = 10^{-10}$.
The same value is also used to enforce conditions (C1) and (C2), i.e., the control parameter in \eqref{eq:delta} is chosen such that
\[c\pdel\leq cE-\eps \quand |F\pm c\pdel| \leq cE \pm F \pm \eps. \]
In Sections \ref{sec_results_sub:instability}-\ref{sec_results_sub:gaussian}, we
study simulations with $c=1$ and neglect the energy equation which is included
in the last two cases from Section \ref{sec_results_sub:5}.
Unless otherwise stated, slope and realizability limiters are always turned on for all DG calculations.
If transformation to characteristic variables for the slope limiter is used, it
will be  explicitly stated. 

\subsection{Two-Beam Instability}
\label{sec_results_sub:instability}
\begin{figure}
\centering
\subfloat[$t=0.6$]{\includegraphics[scale=0.95]{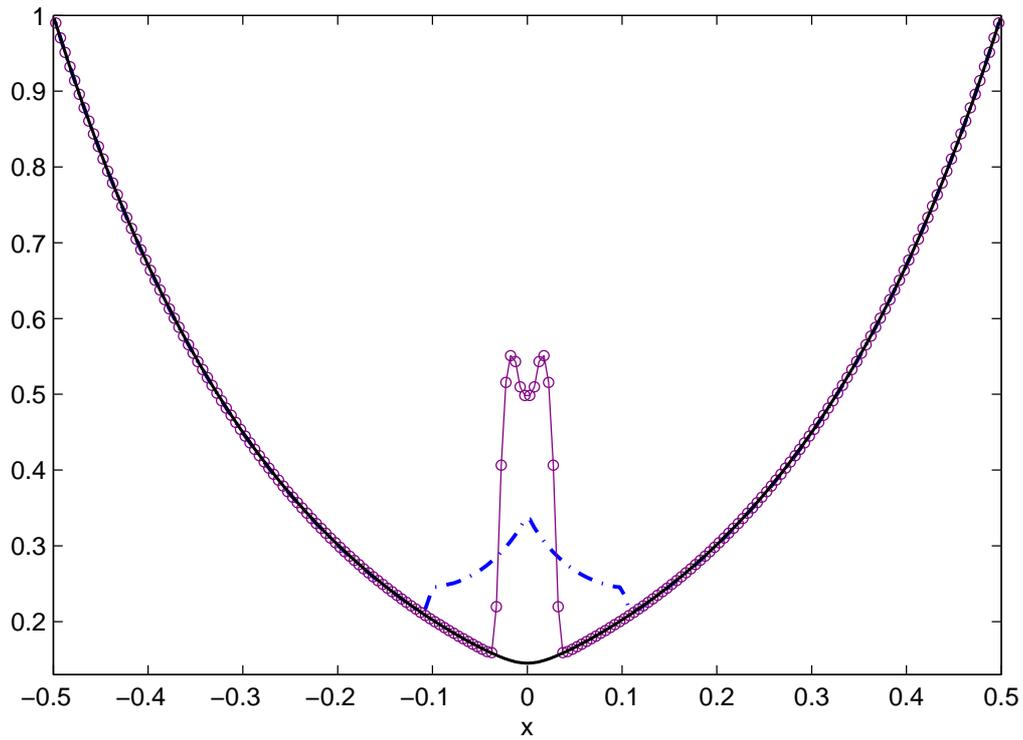}}
\,
\subfloat[$t=3$]{\includegraphics[scale=0.95]{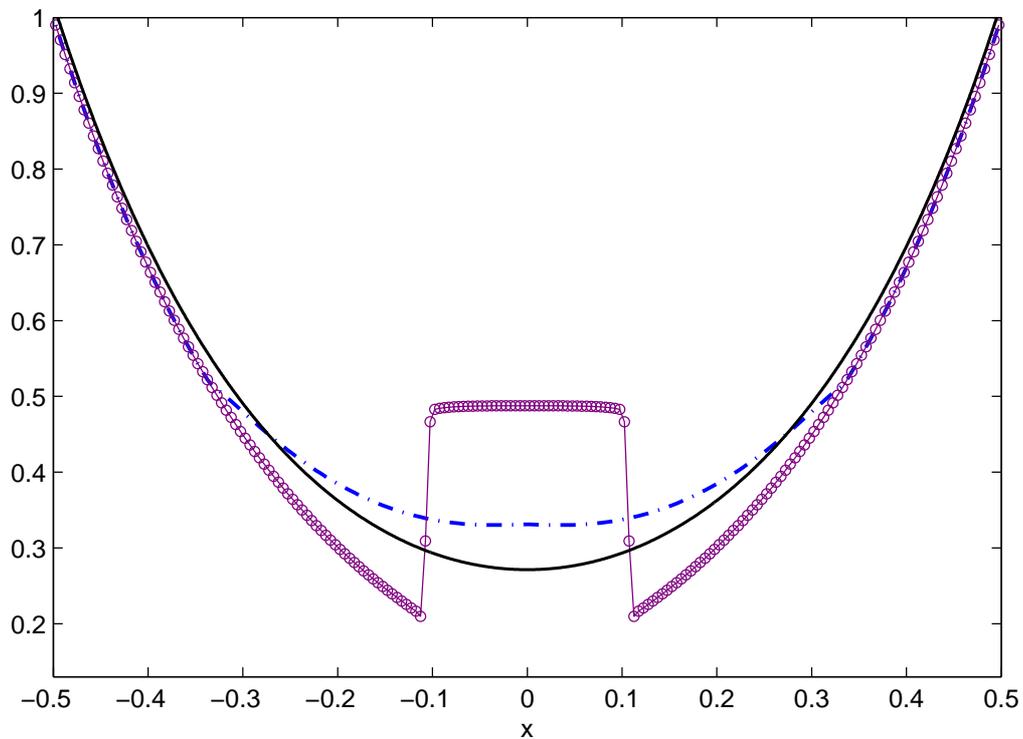}}
\caption{Plots of E for the two-beam instability. $J=200$, $k=2$: $M_1$ (purple circle line), perturbed $M_1$ (blue dash-dot line), transport (black solid line).}

\label{pict:TwoBeams}
\end{figure}

We consider two incoming beams at the boundaries of the domain $[-0.5,0,5]$ and
set $c=1$, $S=0$.\footnote{
The radiation intensity of a beam is a delta distribution in angle at $\mu=1$ (left boundary) and $\mu=-1$ (right boundary) which yields an energy density of $E=1/c$ and a flux density of $F=\pm 1$.
} 
Particles stream from both boundaries in a purely absorbing
material with $\sigma_a=4=\sig{t}$ and meet at $x=0$. 
We avoid getting too close to the boundary of the realizability domain and represent these beams in our moment model using the boundary conditions 
\[ \bu(0,t) =[1,
0.9999]^T, \quad \bu(1,t) =[1, -0.9999]^T, \quad t>0 \] and initial conditions
\[ \bu_0(x) =[2\eps, 0]^T, \quad x\in (-0.5,0,5). \]
For this problem,
coupling to the material is ignored, and the material energy equation is not
included in the simulation.

In \reffig{pict:TwoBeams}, one can observe the formation of a shock in
the $M_1$ solution which persists at the
steady-state.  The perturbed $M_1$ model also develops an unphysical transient
profile in which the particle number jumps in the center of the domain.  While
this artifact persists, the steady state solution ($t=3$) appears continuous.
The steady-state solution also has noticeable kinks in the at $x\approx\pm 0.3$.
For comparison, discrete ordinates solutions are plotted for which $256$
discretization points in angle and $1000$ points in space are used. The
perturbed $M_1$ is throughout closer to the transport solution.

\begin{rem}
Precise explanations for the occurrence of shocks and kinks in the perturbed
$M_1$ solution require an additional analysis. For example, an
explanation for the formation of shocks in the $M_1$ model is given
\cite{Brunner-Holloway-2001}. However, such an analysis goes beyond the
purpose of this paper and will be postponed to future work.
\end{rem}

\subsection{Source-Beam Problem}
\label{sec_results_sub:source_beam}

\begin{figure}[h!]
\centering
\subfloat[$t=0.5$]{\includegraphics[scale=0.5]{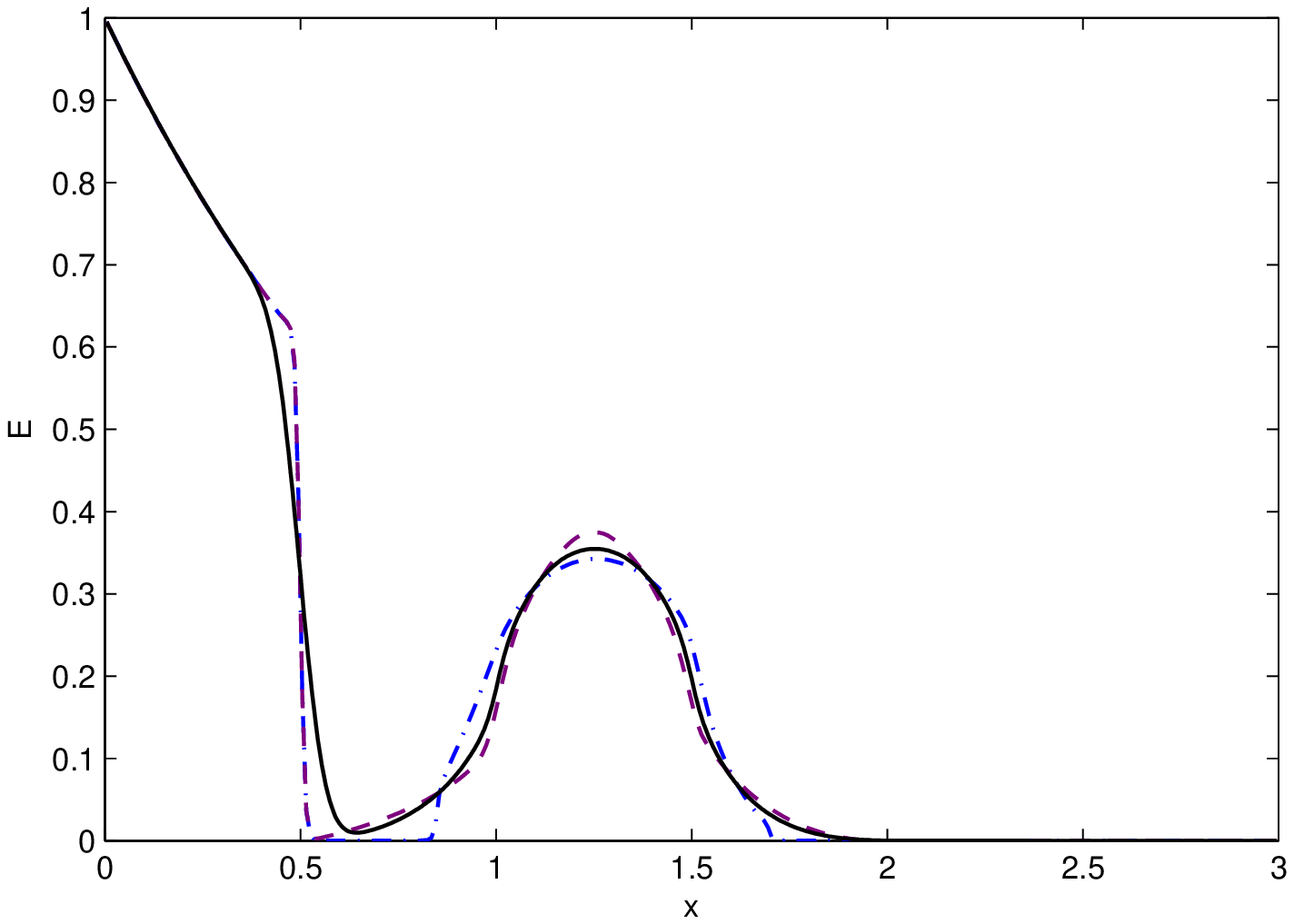}}
\,
\subfloat[$t=1$]{\includegraphics[scale=0.5]{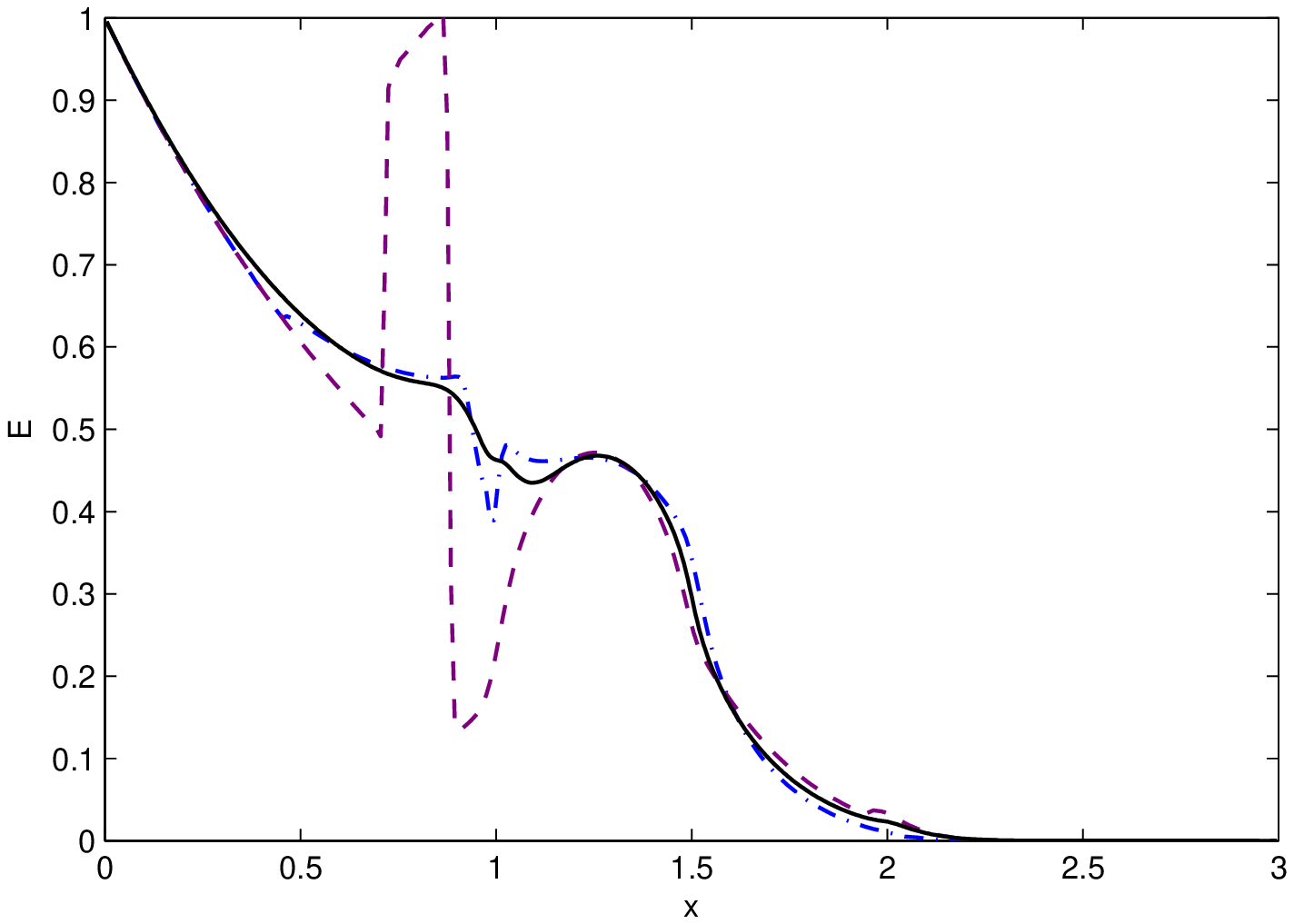}} \\
\subfloat[$t=2$]{\includegraphics[scale=0.5]{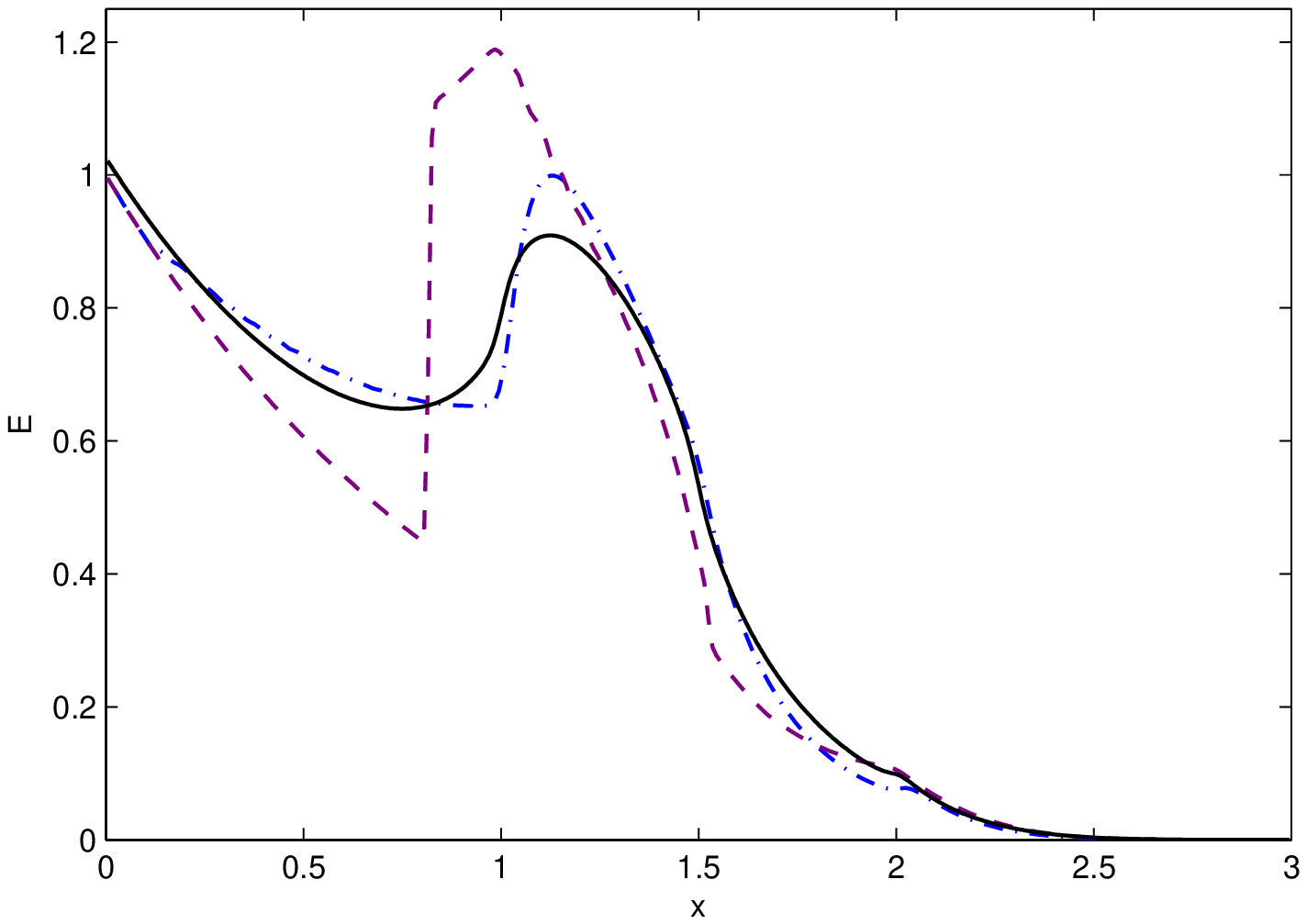}}
\,
\subfloat[$t=4$]{\includegraphics[scale=0.5]{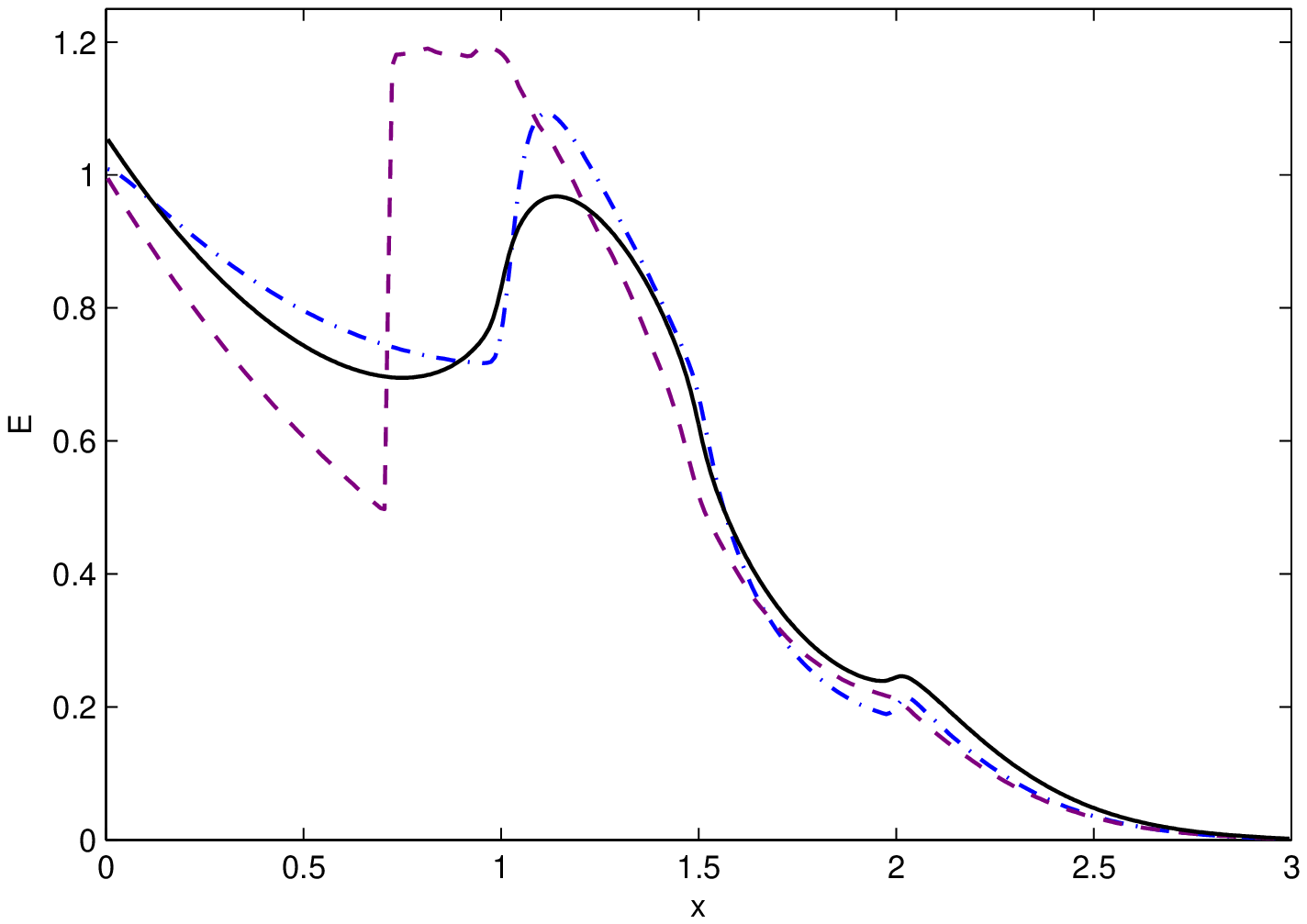}}
\caption{Plots of E for the Source-beam problem. $J=300$, $k=2$: $M_1$ (purple dashed line), perturbed $M_1$ (blue dash-dot line), transport solution (black solid line)}
\label{pict:HeatedWall}
\end{figure}%
In this problem, an incoming beam \[ \bu(0,t) =[1, 0.9999]^T, \quad t>0, \] on
the left boundary of the domain $[0,3]$ hits an isotropic source $S=1/2$
generating particles on the interval $1\leq x<1.5$.  In order to avoid
complications of spatial discontinuities in the fluxes
\cite{ZhaWonShu06,ZhaLiu05,LonTemJin95}, we smooth the source and
material cross sections, which enter into the perturbative components of the
flux. The source $S$ is smoothed
at the end points $x=1$ and $x=1.5$
\begin{equation}
S =\begin{cases}
 			\frac14(1+p_H(\frac{x-1}{\Delta})), & 1-\Delta \leq x \leq 1+\Delta \\
			\frac12, & 1+\Delta<x<1.5-\Delta\\
			\frac14(1-p_H(\frac{x-1.5}{\Delta})), & 1.5-\Delta \leq x \leq 1.5+\Delta \\
 			0, & \tx{else}
            \end{cases}
\end{equation}
Similarly, we design the material properties with the cross sections:
\begin{align}
\sig{a} &=\begin{cases}
			1, & 0\leq x < 2-\Delta \\
 			(1-p_H(\frac{x-2}{\Delta}))/2, & 2-\Delta \leq x \leq 2+\Delta \\
 			0, & \tx{else}
            \end{cases}, \quand \\
\sig{s} &= \begin{cases}
			1+p_H(\frac{x-1}{\Delta}), & 1-\Delta \leq x \leq 1+\Delta \\
 			2, & 1+\Delta < x < 2-\Delta \\
 			2+4(1+p_H(\frac{x-2}{\Delta})), & 2-\Delta \leq x \leq 2+\Delta \\
 			10, & 2+\Delta < x \leq 3 \\
 			0, & \tx{else.}
            \end{cases}
\end{align}
The function $p_H$ is a Hermite polynomial of order 10 with $p_H(\pm 1)=\pm 1$
and $p^{(k)}_H(\pm1)=0$ for $k=1,2,3,4$. If $p_H$ is extended by $\pm 1$
respectively, it is a $C^4$ function. 
The material property functions are illustrated in \reffig{pict:SourceBeam_material}.

\begin{figure}[h!]
\centering
\includegraphics[scale=0.42]{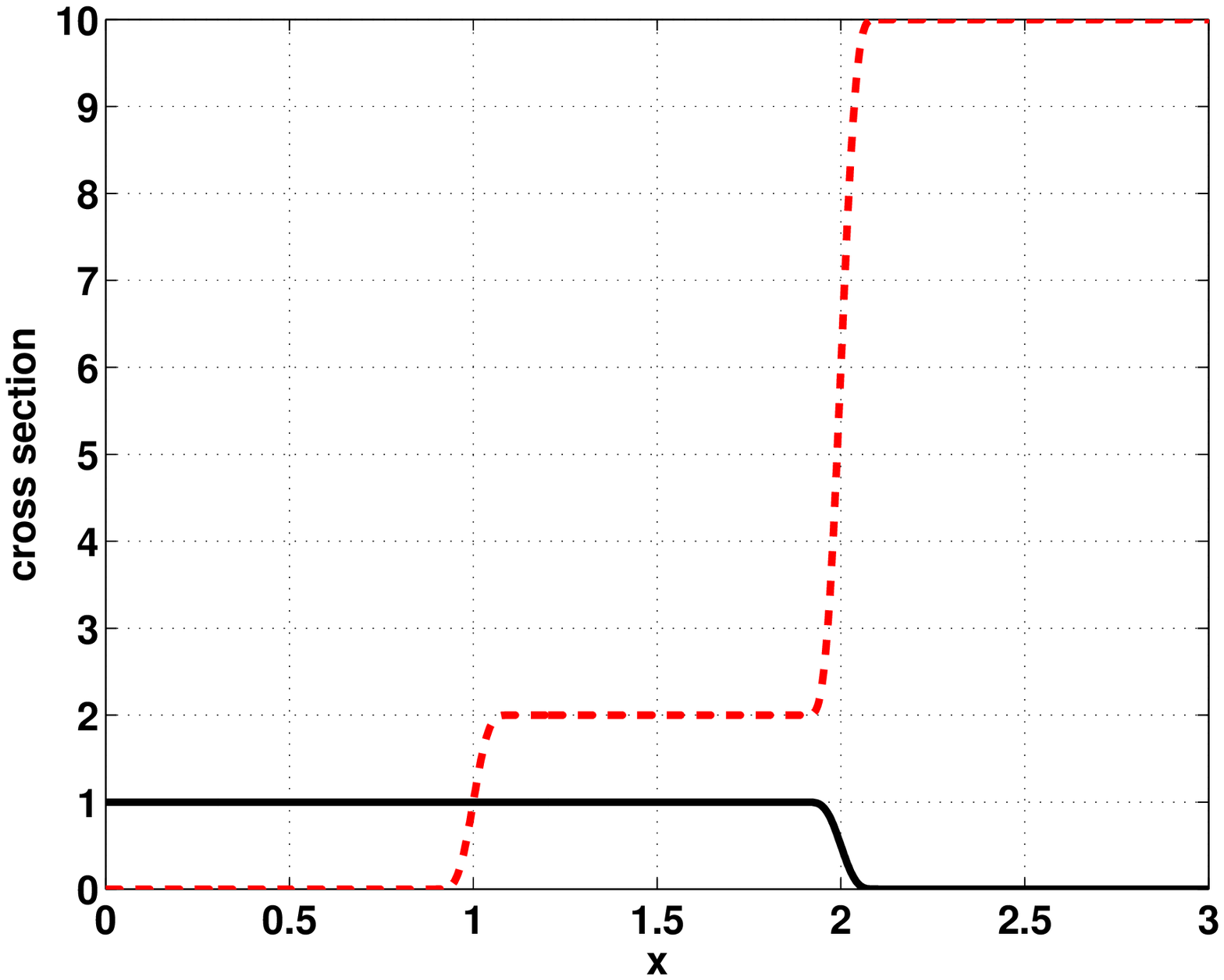}
\,
\includegraphics[scale=0.42]{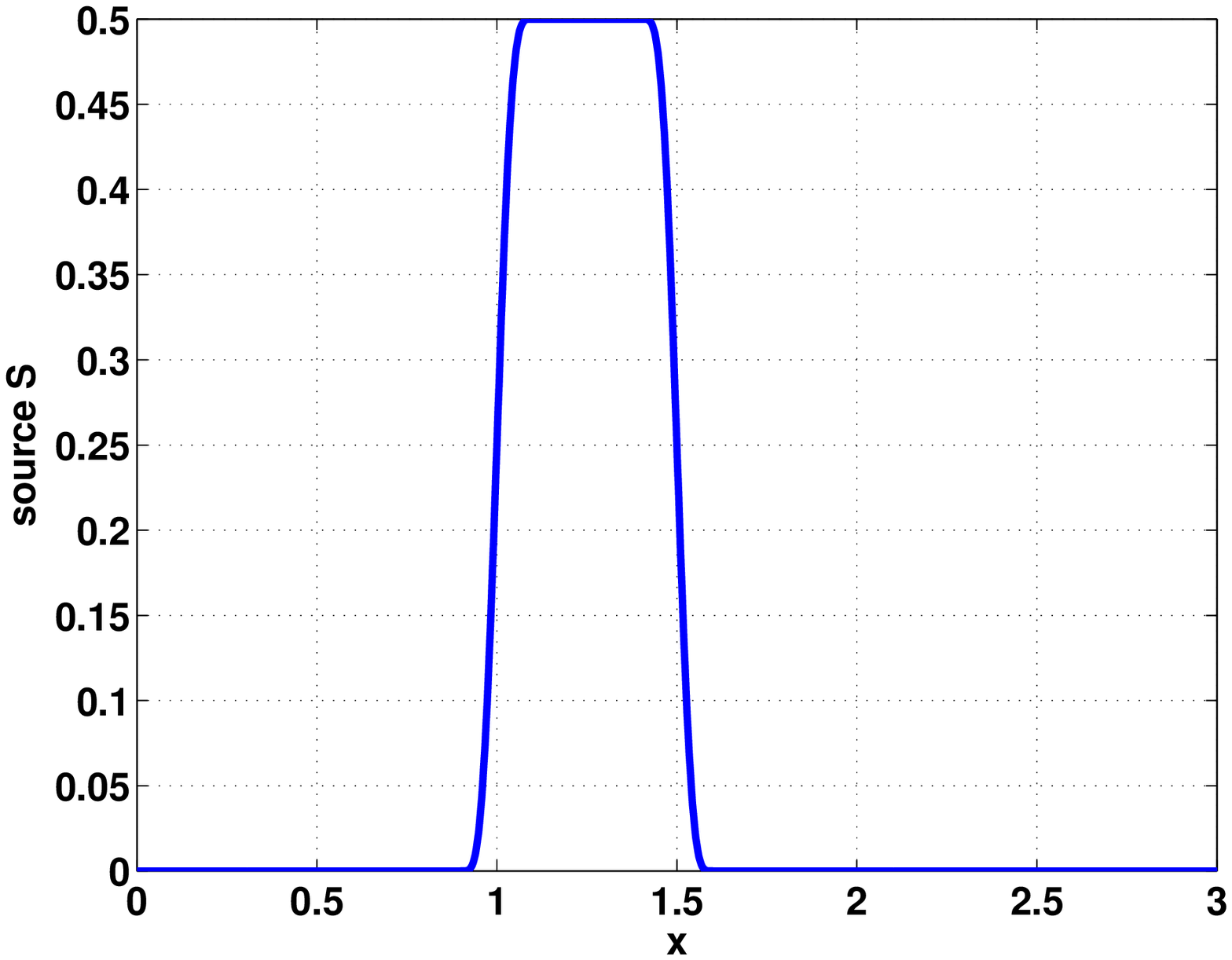}
\caption{Source-beam problem. Material properties with $\Delta=0.05$: $\sigma_s$ (red dashed line), $\sigma_a$ (black solid line), $S$ (blue solid line)}
\label{pict:SourceBeam_material}
\end{figure}%

On the right boundary, particles are absorbed and zero Dirichlet conditions 
\begin{equation}
\bu(3,t) =[\varepsilon, 0]^T, \quad t>0, 
\end{equation}
are set.
Initially, there are (almost) no particles in the system, i.e.,
\begin{equation}
\bu(x,0) =[\eps, 0]^T, \quad x\in  (0,3). 
\end{equation}
The value of $c$ is again set to one.

The perturbed $M_1$ results are compared to $M_1$ and transport solutions.
Classic $M_1$ calculations are performed using the DG method from
\cite{OlbHauFra11}, with the same computational parameters as the $PM_1$
model and slope limiting performed in the characteristic variables.  The
transport solution is computed using the discrete ordinates method with $600$ spatial cells and $256$ discrete
angles.

One can observe in \reffig{pict:HeatedWall} that as time increases, particles
entering from the left boundary encounter the source in the interior.  As this
happens the $M_1$ profile diverges from the transport solution. Even as
steady state is achieved at $t=4$, there is a large difference for $x\leq 1$.
The $PM_1$ profile agrees much better with the transport solution.

\subsection{Gaussian Source}
\label{sec_results_sub:gaussian}
The next test case simulates particles with an initial energy density that is
a Gaussian distribution in space and a zero energy density flux:
\[\bu(x,0) =\left [\frac{1}{\xi\sqrt{2\pi}} \, e^{ -\frac{x^2}{2\xi^2}}, 0
\right ]^T, \q \xi=0.1, \quad x\in  (-L,L). \]
Periodic boundary conditions on $[-L, L]$ are prescribed where $L=t_{\tx{final}}+1$.
The computational domain is always chosen large enough to ensure that a
negligible number of particles reaches the boundaries. No internal source is
assumed (so $S=0$), and the medium is purely scattering with
$\sig{s}=\sig{t}=1$. The velocity $c$ is also set to one and the material energy
equation \eqref{eq:material_energy} is neglected. All DG results are computed
with $h=0.01$ and polynomial degree $k=2$.  For comparison, discrete ordinates
solutions of the transport equations are obtained with $h=0.005$ and $128$
angular points.

\reffig{pict:Gaussian} displays the solutions at $t_{\tx{final}}=1,2,3,10$. The
$M_1$ model gives the expected wave effects that are washed out at larger
times. These effects do not occur in the perturbed $M_1$ results. However,
the $PM_1$ forms Gaussian bell that are higher and more narrow than the
benchmark solution. At lower times, their maximum propagation speed is roughly
half the correct velocity. Nevertheless, at $t_{\tx{final}}=10$ the
front of the $PM_1$ model catches up with the reference solution, at which
point all three models agree reasonably well.

\begin{figure}[h!]
\centering
\subfloat[$t_{\tx{final}}=1$]{\includegraphics[scale=0.52]{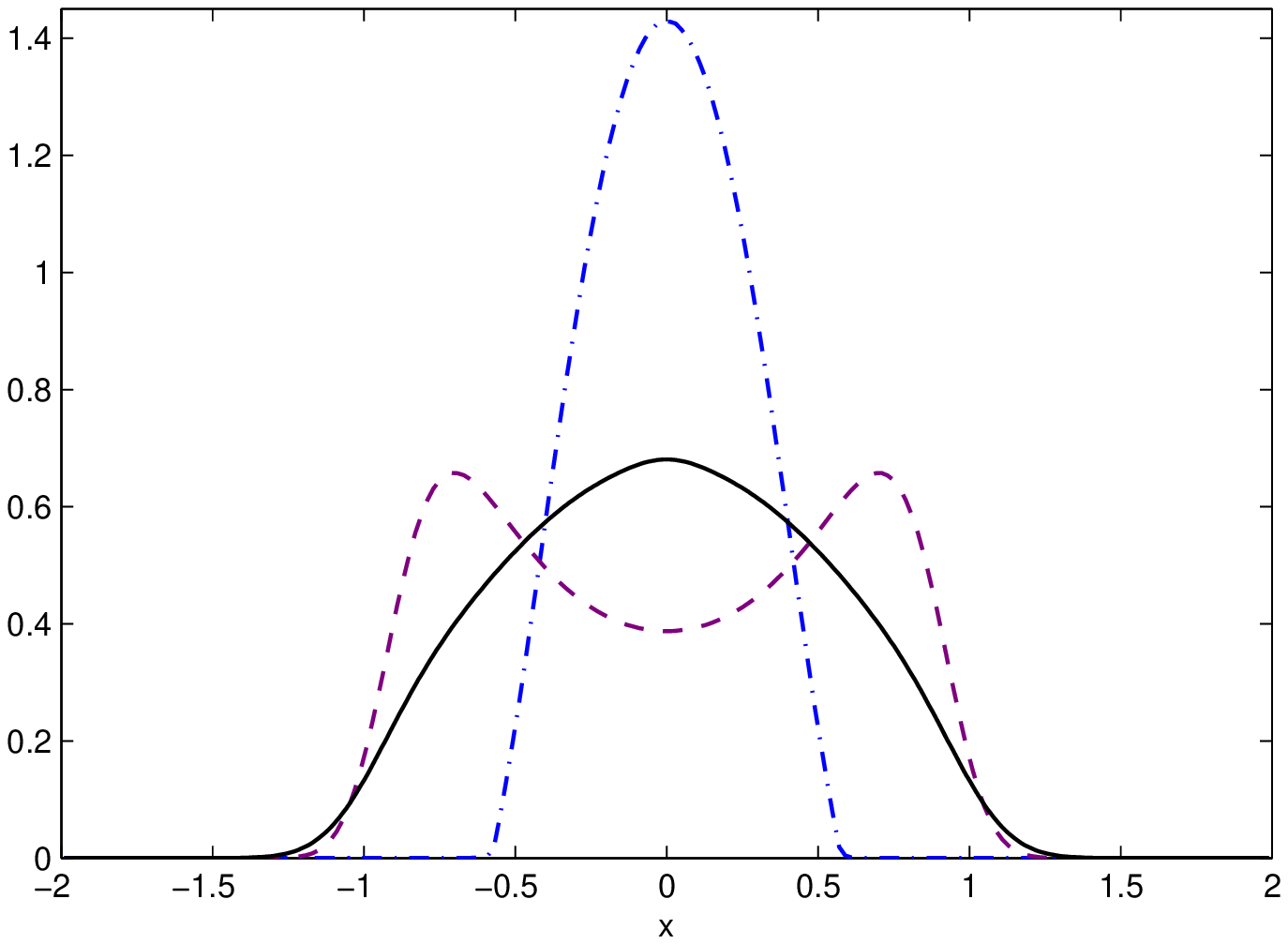}}
\,
\subfloat[$t_{\tx{final}}=2$]{\includegraphics[scale=0.52]{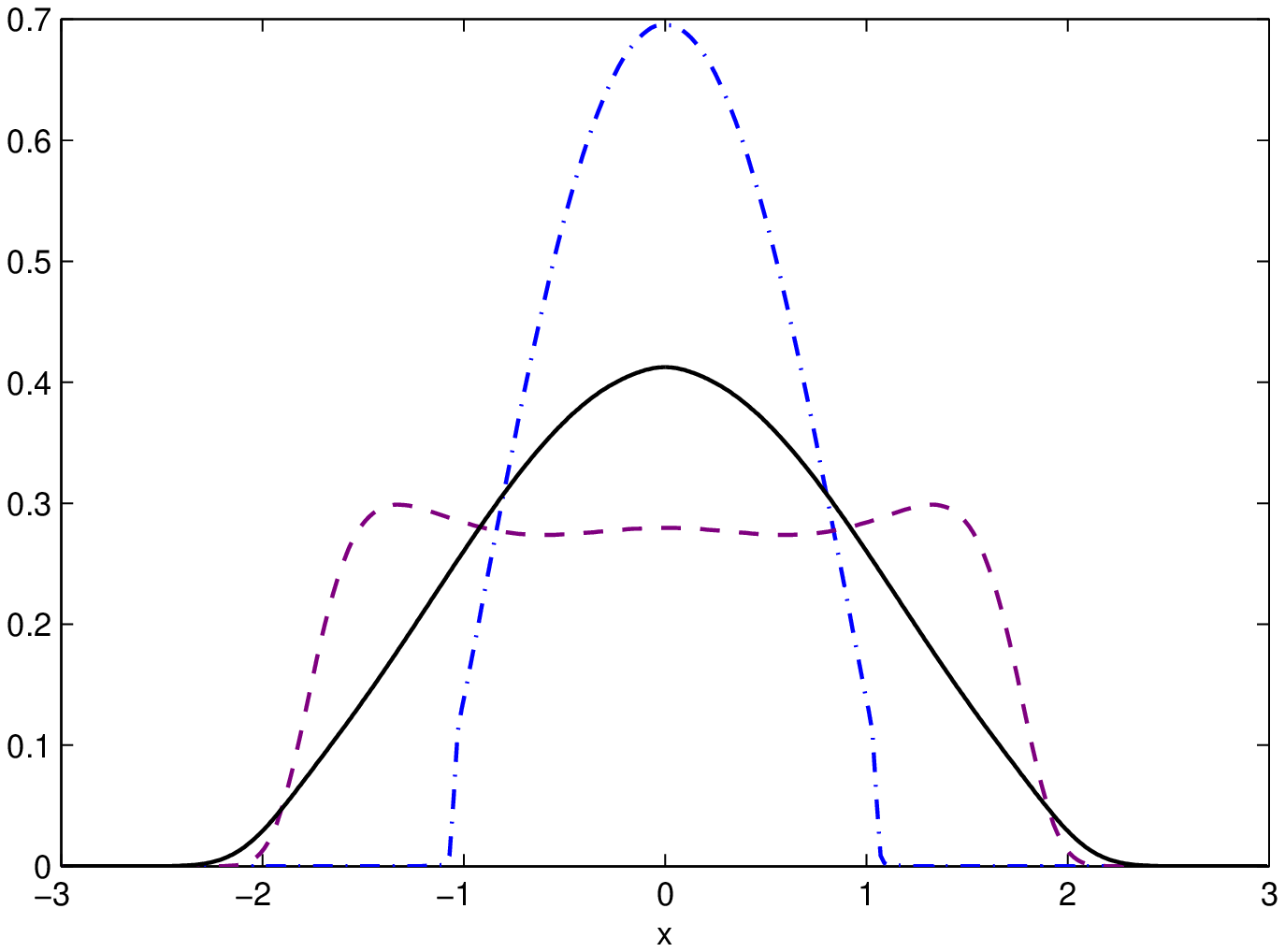}} \\
\subfloat[$t_{\tx{final}}=3$]{\includegraphics[scale=0.52]{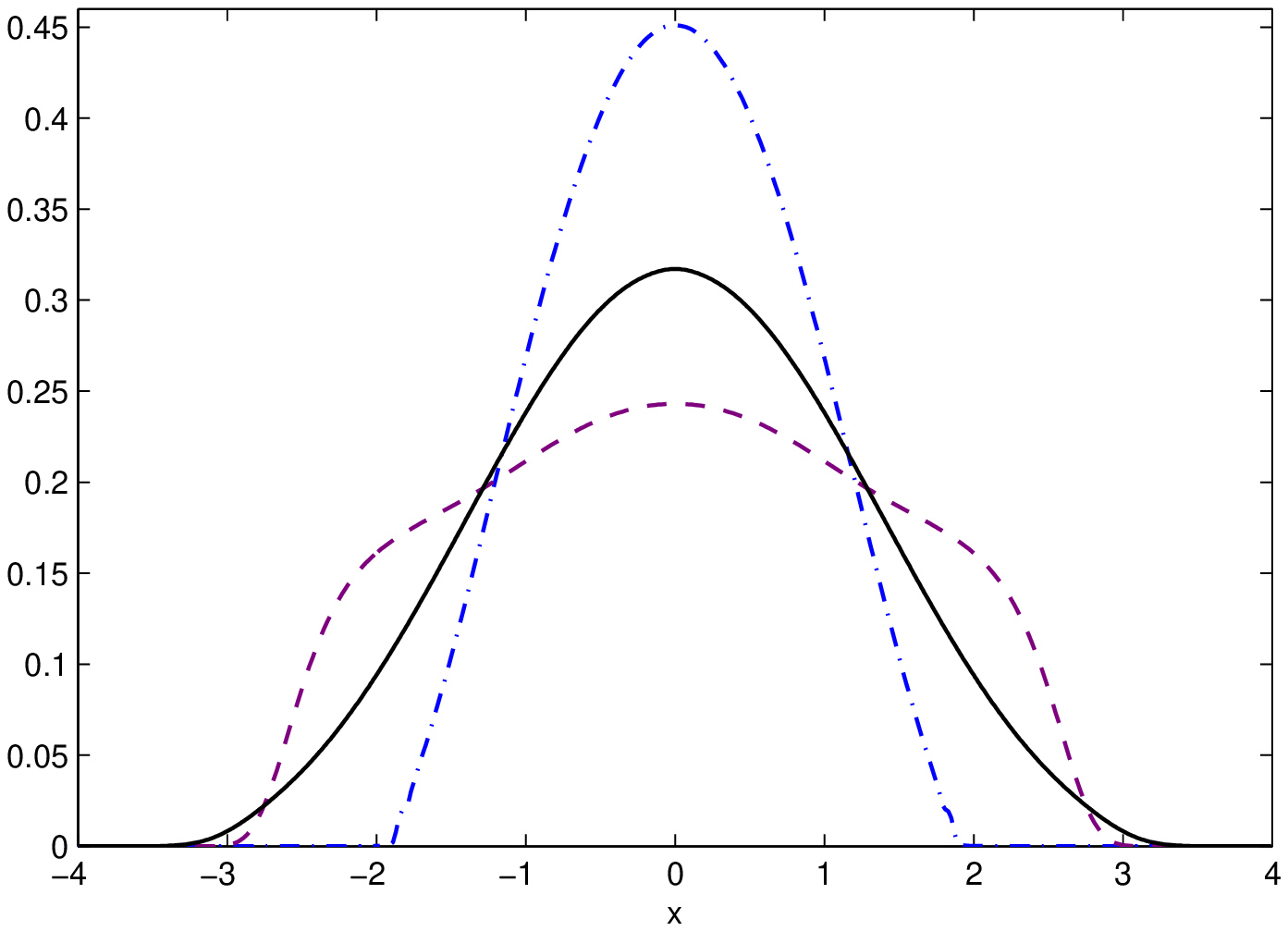}}
\
\subfloat[$t_{\tx{final}}=10$]{\includegraphics[scale=0.52]{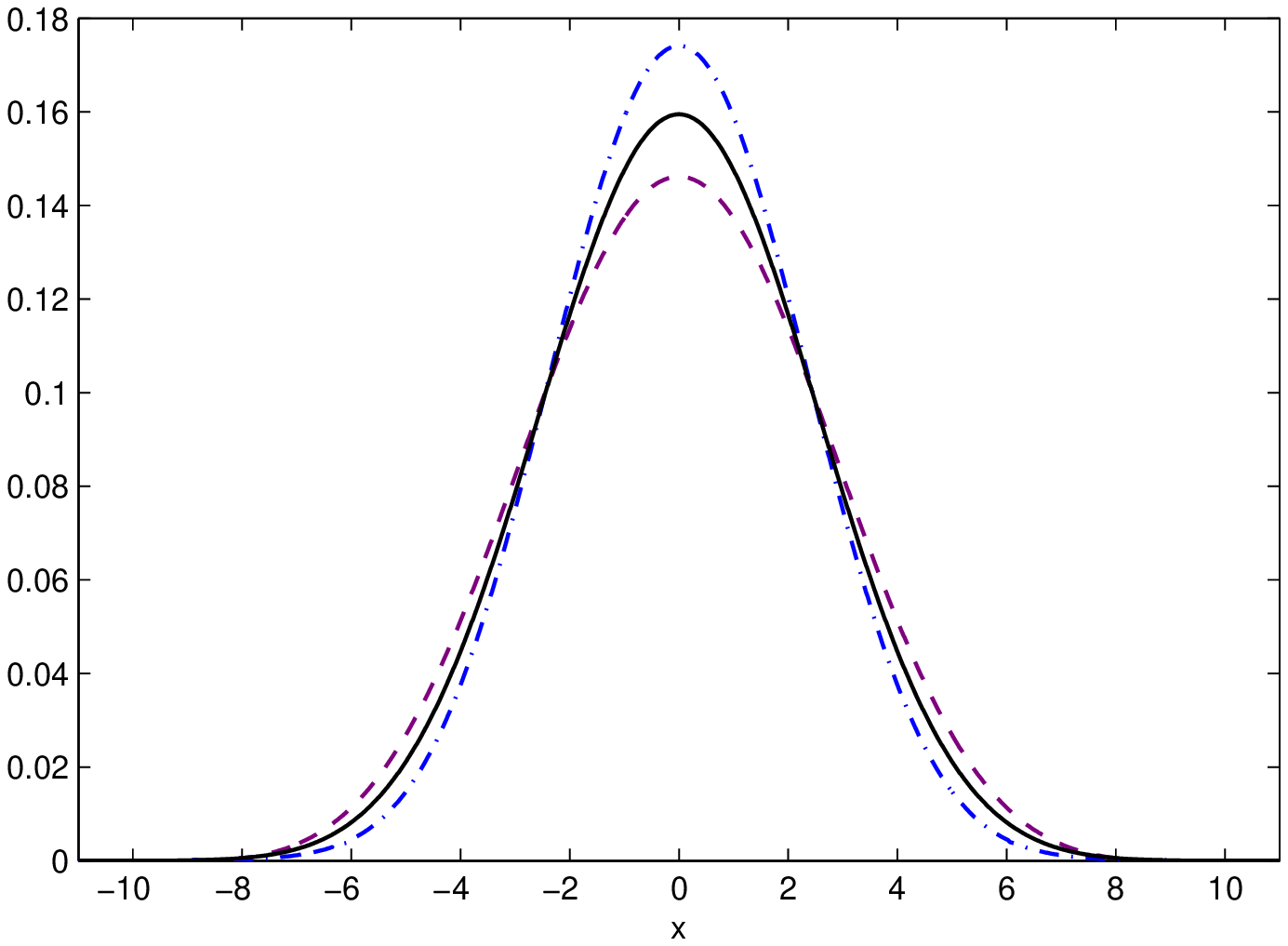}}
\caption{Plots of E for the Gaussian source. $J=100\, (t_{\tx{final}}+1)$, $k=2$: $M_1$ (purple circle line), perturbed $M_1$ (blue dash-dot line), transport (black solid line).}
\label{pict:Gaussian}
\end{figure}%

The perturbation $\tilde{\psi}$ from Section~\ref{sec:PEB} is related
to the difference between the $M_1$ and transport solution.
\reffig{pict:Gaussian} indicates that this quantity is highly time-dependent.
Additionally, the spatial gradient of $\tilde{\psi}$ is large at shorter times.
Hence, this numerical example violates the smallness assumptions made in the
derivation of the perturbed $M_1$ model in Section~\ref{sec:PEB}.  Thus the
lack of accuracy is not surprising.

\subsection{Including the Material Energy Equation}
\label{sec_results_sub:5}
The next two examples involve \eqref{eq:transport} coupling to the energy
equation \eqref{eq:material_energy}. The linearized Marshak wave problem from
\cite{SuOls97} is analyzed first and then a Marshak wave with material
parameters taken from \cite{OlbHauFra11}.

\subsubsection{Smoothed Su-Olson's Benchmark Problem}
\label{sec_results_sub:su_olson}

\begin{figure}[h!]
\centering
\subfloat[$t=0.1$]{\includegraphics[scale=0.5]{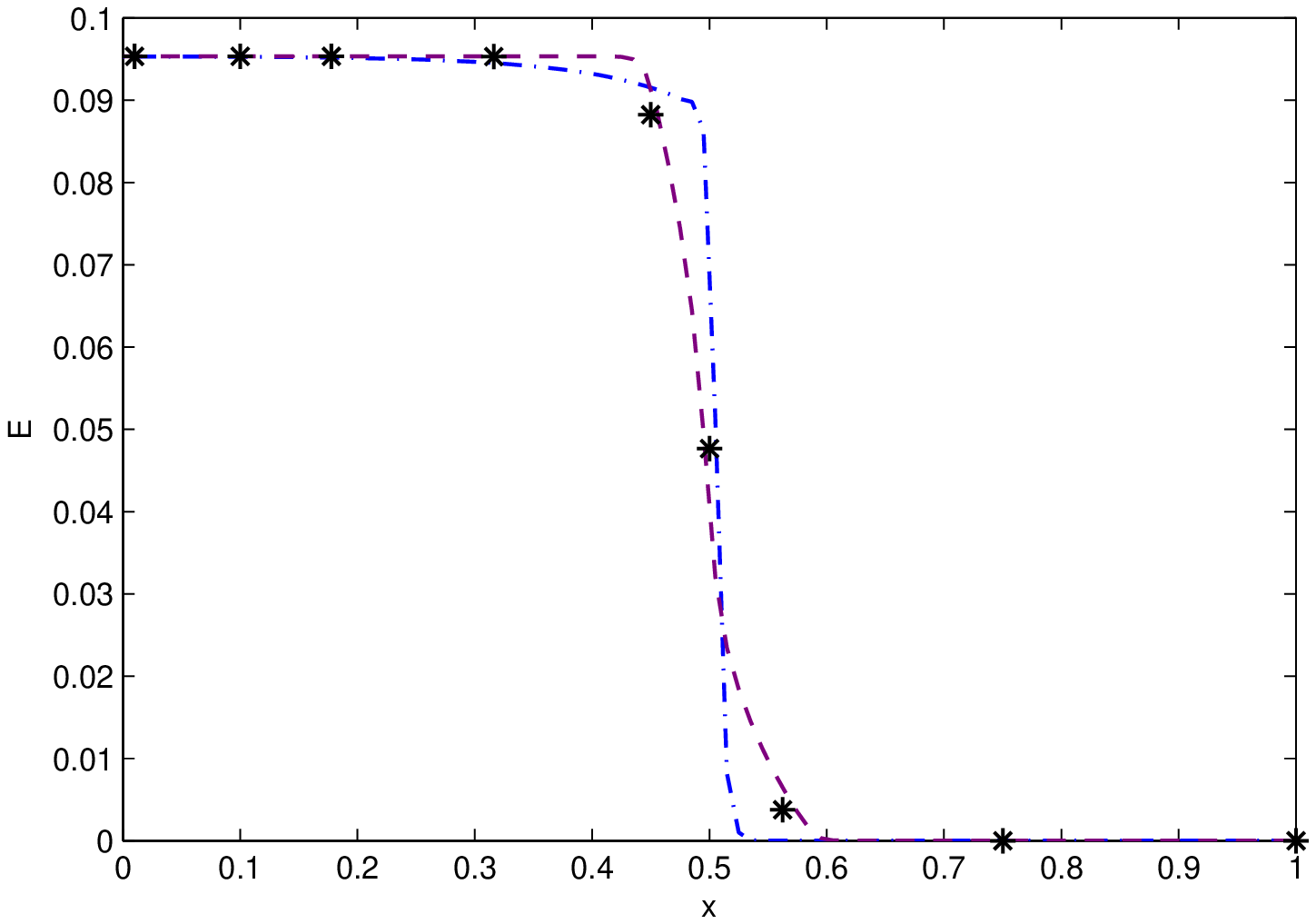}}
\,
\subfloat[$t=0.31623$]{\includegraphics[scale=0.5]{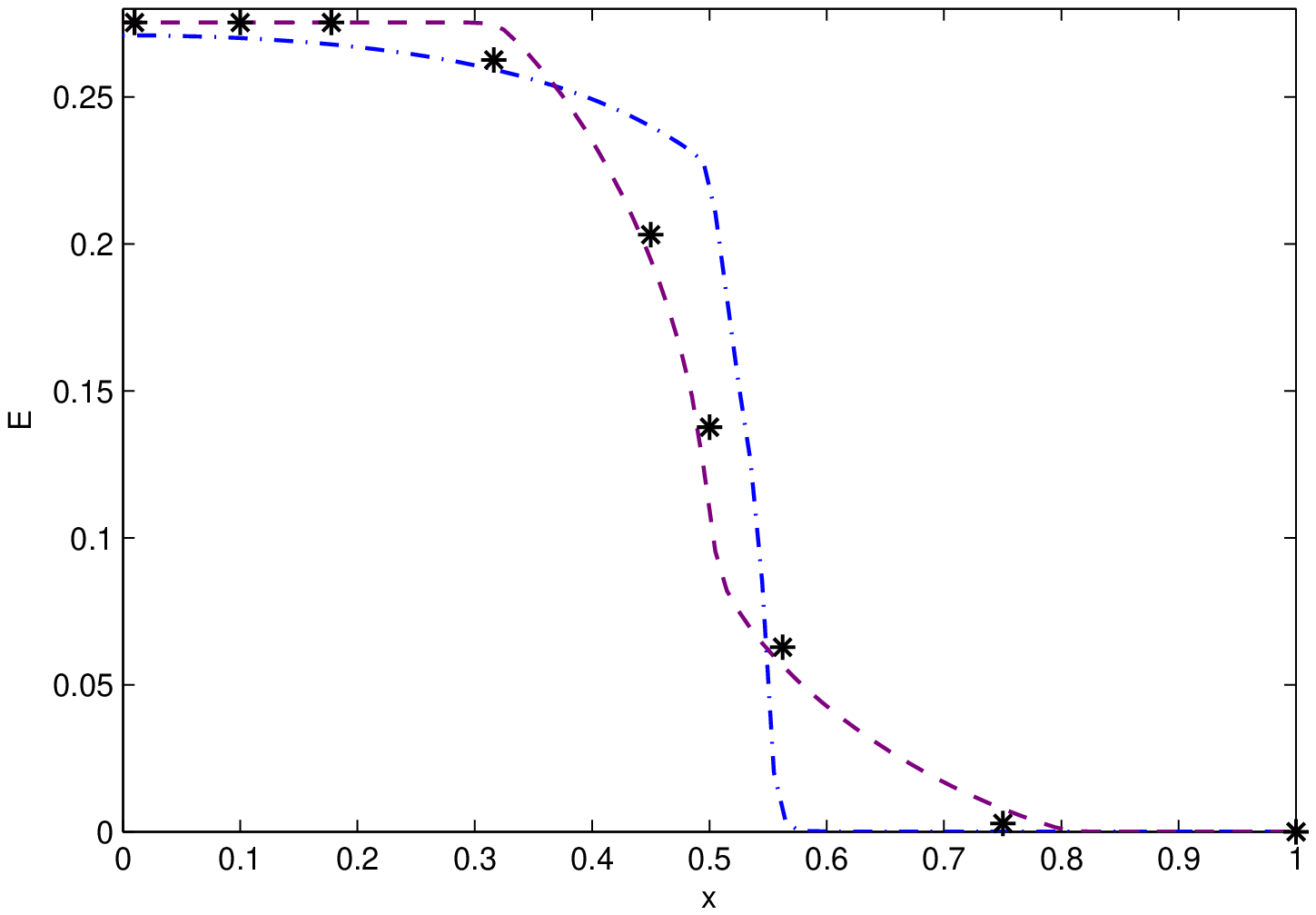}} \\
\subfloat[$t=1$]{\includegraphics[scale=0.5]{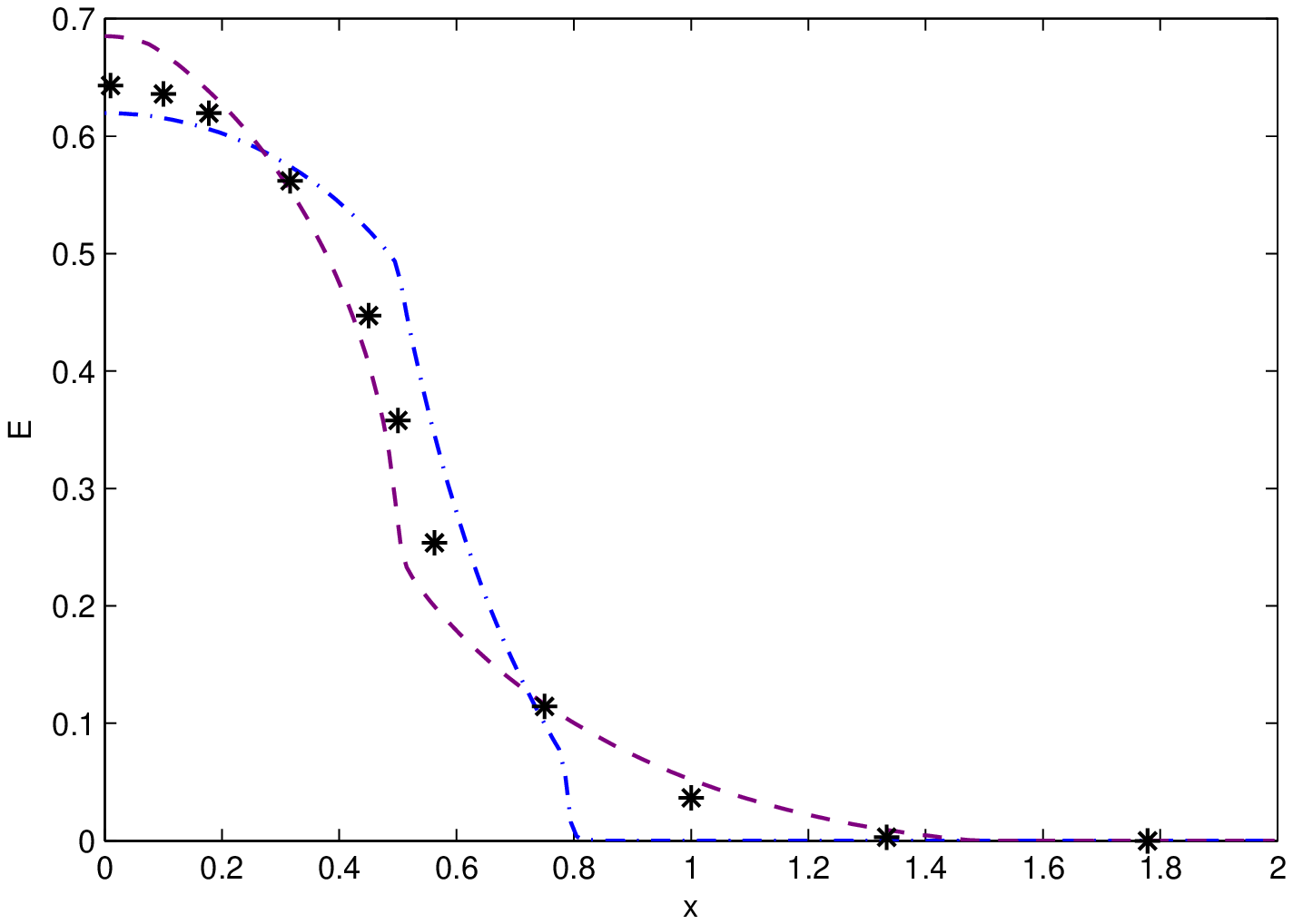}}
\,
\subfloat[$t=3.16228$]{\includegraphics[scale=0.5]{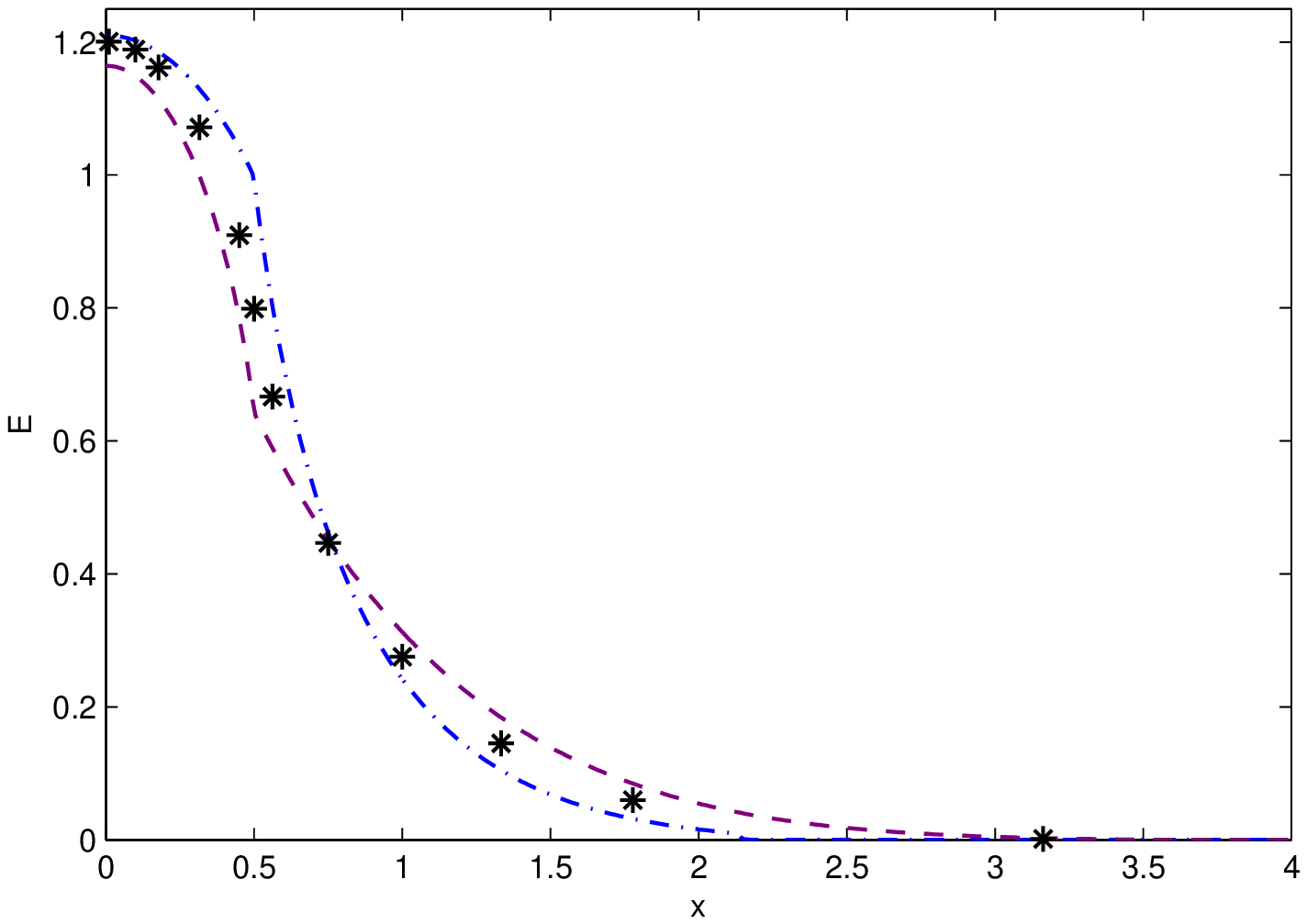}}
\caption{Plots of E for the Su-Olson Problem. $J=200$, $k=2$: $M_1$ (purple dashed line),
perturbed $M_1$ (blue dash-dot line), semi-analytic (black asterisk).}
\label{pict:SuOlson}
\end{figure}%

In \cite{SuOls97}, tabulated data is provided for analytic solutions to a
linearized Marshak wave problem, which serves as a validation of numerical
algorithms in the radiative transfer community. In particular, this
semi-analytic benchmark is also compared to diffusion-corrected $P_N$
approximations in \cite{SchFraLev11}. It is therefore of interest to study
solutions of the perturbed $M_1$ model to this problem.

We compute approximations to \eqref{eq:transport} and
\eqref{eq:material_energy} in slab geometry with the following physical data
\[C_v = T^3, \q c=1, \q \sig{a}=1=\sig{t}. \]
As in the source-beam problem, we seek to avoid the discontinuous material
properties in \cite{SuOls97} by introducing a smoothed version:
\[ S(x,t) = \begin{cases}
				(1+p_H(\frac{x+0.5}{\Delta}))/2, & -0.5-\Delta \leq x \leq -0.5+\Delta, \\
				1, & -0.5+\Delta < x < 0.5-\Delta, \\
				(1-p_H(\frac{x-0.5}{\Delta}))/2, & 0.5-\Delta \leq x \leq 0.5+\Delta, \\
				0, & \tx{else,}
                \end{cases}
\]
where $p_H$ is the Hermite polynomial described in
Section~\ref{sec_results_sub:source_beam}. However, we still 
compare our numerical results to the semi-analytic solutions from \cite{SuOls97}
because the length of the smoothing window $2\Delta=0.02$ is relatively small.

Initially, the medium is cold and there is no radiation: \[\psi(x,\mu,\nu,t=0) =
0 \quand  T(x,t=0) = 0.\] Additionally, zero boundary conditions are enforced on
an infinite domain: \[ \lim_{x\rightarrow\pm\infty} \psi(x,\mu,\nu,t) = 0 \quand
 \lim_{x\rightarrow\pm\infty} T(x,t) = 0. \] In practice, we impose periodic
boundary conditions on a large domain $[-L,L]$ where $L=\lfloor t_{\rm{final}}
\rfloor+1$.

Solutions at different times are provided in \reffig{pict:SuOlson} for the half
plane $x\geq 0$. A grid size of $h=0.01$ and polynomial degree of $k=2$ are
chosen for all DG solutions. Classic $M_1$ computations are slope limited in the
characteristic variables. They are throughout close to the semi-analytic
results. However at earlier times, the perturbed $M_1$ solutions have larger
slopes at $x\approx 0.5$. There, the perturbed $M_1$ model yields
larger deviations from the reference. Only at $t=3.16228$ solutions from both
models are close to each other as well as to the semi-analytic points.


\subsubsection{Thin Marshak Wave}
\label{sec_results_sub:marshak}
In this problem, incoming radiation is prescribed on the left boundary by
well-posed boundary conditions,
 \begin{align*}
  T(0,t) &= 1, \quad T(1,t)=0, \quad t>0, \\
 \bu(0,t) &=[T(0,t)^4 a, 0.8\cdot T(0,t)^4 ac]^T, \quad \bu(1,t) =[2\veps, 0]^T,
 \quad t>0,
 \end{align*}
and the material is assumed to be purely absorbing:
\[ \sig{a} (T) = \frac{1}{(T +0.5)^3} \frac{\tx{keV}^3}{\tx{cm}}, \quad
 \sigma_s=0, \q S=0. \]
The physical constants are given by
\begin{align*}
c &= 3\cdot 10^{10} \, \tx{cm/s} & \tx{speed of light}, \\
a &= 1.372 \cdot 10^{14} \, \tx{erg/(cm}^3 \tx{keV}^4) & \tx{radiation constant}, \\
C_v &= 3\cdot 10^{15} \, \tx{erg/(cm}^3 \tx{keV)}  & \tx{heat capacity},
\end{align*}
which implies units of cm$^{-1}$ for cross sections and keV for temperature $T$.
Initially, the material is cold
\begin{align*}
 T_0(x) & = 5\cdot 10^{-4}, \quad x\in  (0,1), \\ 
 \bu(x,0) &=[T_0(x)^4 a, 0]^T, \quad x\in  (0,1).
 \end{align*}

Due to above incoming radiation on the left boundary, radiation propagates
through the medium from left to the right. The material temperature $T$
(\reffig{pict:ThinMarshak}a) and the energy density $E$
(\reffig{pict:ThinMarshak}a) decay smoothly to zero. The $M_1$ and \ptm models
yield very similar solutions. For comparison, a reference solution is computed
using a $P_{99}$ model that is calculated with the DG method from
\cite{McCEvaLow08}. The simulation of the $P_{99}$ model uses linear elements
and $800$ spatial cells.  The reference solution shows a much stronger decay in
the energy and material temperature.
\begin{figure}
\centering
\subfloat[Temperature at $t=0.1$ns.]{\includegraphics[scale=0.95]{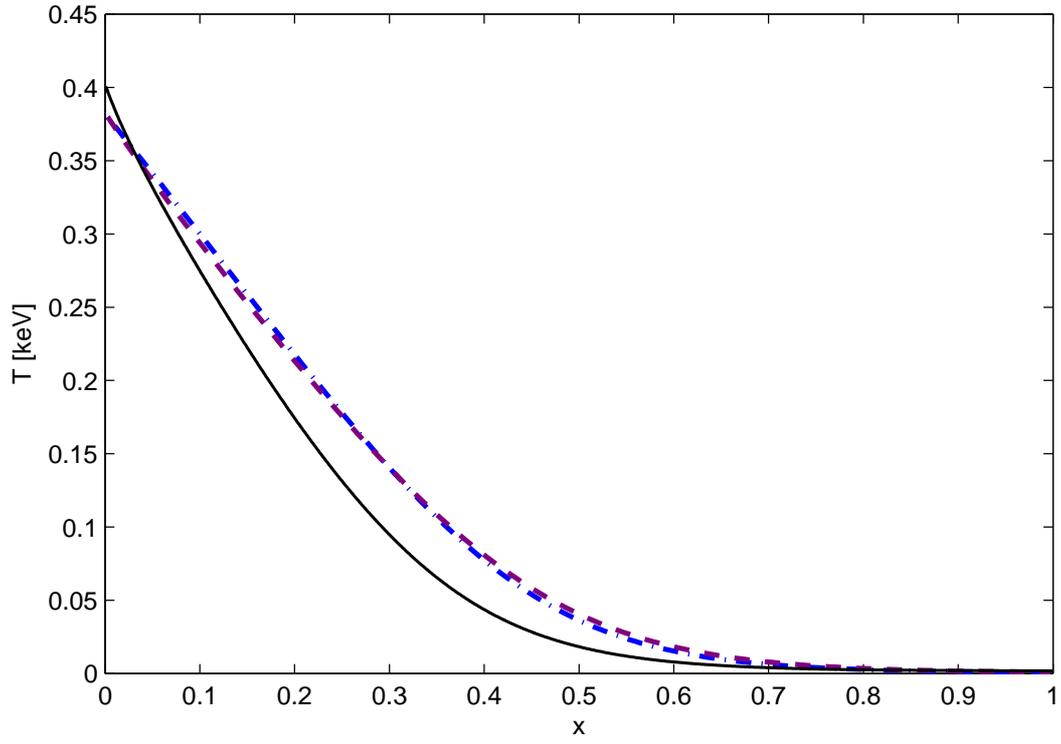}}
\,
\subfloat[Scaled scalar flux $a^{-1}E$ at $t=0.1$ns.]{\includegraphics[scale=0.95]{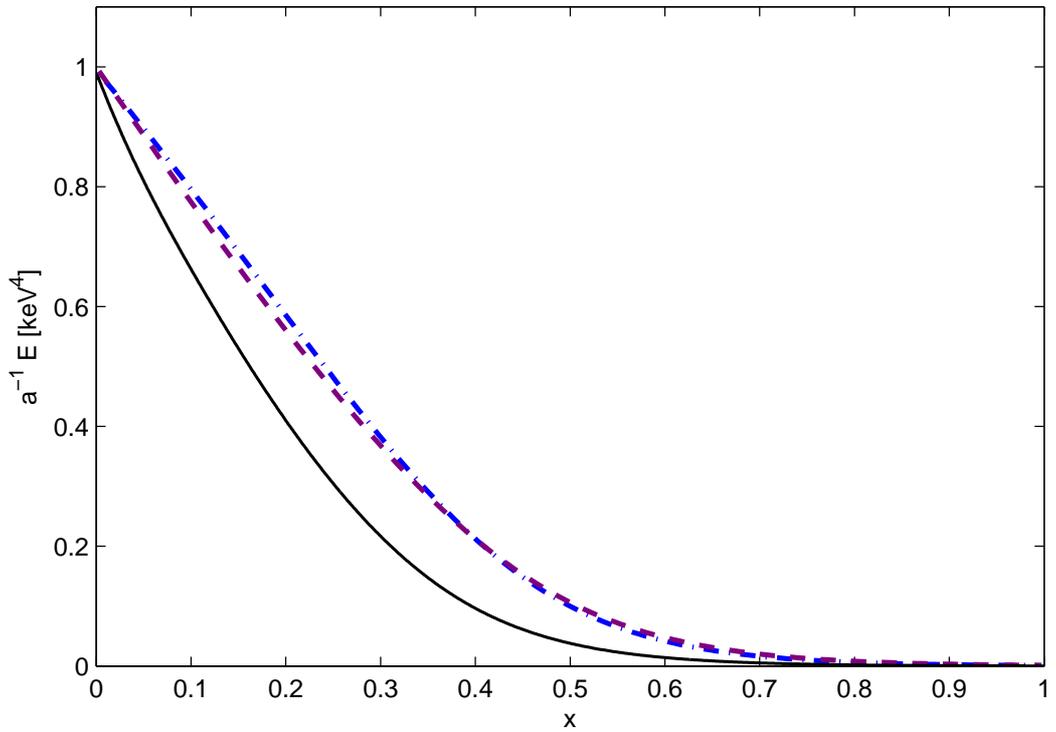}}
\caption{Thin Marshak wave. $J=160$, $k=2$: $M_1$ (purple dashed line), perturbed $M_1$ (blue dash-dot line), $P_{99}$ (black solid circle line).}
\label{pict:ThinMarshak}
\end{figure}
\section{Discussion and Conclusions}
\label{sec:disc}

In this paper, we have derived a hierarchy of closures based on perturbations
of well-known entropy-based closures.  The derivation has been done in the
context of grey photon transport.  Our derivations reveals final equations
containing an additional convective and diffusive term which are added to the
flux term of the standard closure. This is different to perturbations to
standard $P_N$ closures \cite{SchFraLev11} which only gain a diffusive
component.

For the first member of the hierarchy, the \ptm model, we compute explicit
formulas for all terms.  The resulting system of equations is a
convection-diffusion system which, for slab geometries, is discretized by using
a Runge-Kutta discontinuous Galerkin method. By introducing a special limiter
and an additional control parameter to modify the pressure term, we ensure that
cell averages of the moments satisfy the important realizability property
\eqref{eq:real}.

We perform simulations to compare qualitative results with the $M_1$ model and
with highly-resolved discretizations of the original transport equation.
Improvements to the standard $M_1$ model are observed in cases where unphysical
shocks develop in the classical $M_1$ model. However, for problems with
continuous solutions, there is little or no improvement; and in some cases, the
$M_1$ model perform marginally better.

Finally, we discuss some open problems in this framework which might be addressed in future:
\begin{itemize}
\item{
Moment systems from entropy-based closures are known to be hyperbolic and
satisfy a local dissipation law \cite{Levermore-1996,Dubroca-Feugas-1999},
but such results are not yet known for the perturbed models.  The partial result
in Proposition~\ref{prop:dissipation} confirms that the additional diffusive
term dissipates the entropy. Moreover, neglecting the diffusion contribution the
system indeed becomes hyperbolic for the special case of the $M_1$ model in
one dimension (i.e., for slab geometry).}
\item{
Standard issues in analysis such as existence and uniqueness of solutions have
yet
to be investigated for the $PM_1$ model or for PEB models in general.}

\item{
In Section \ref{sec_results_sub_sub:set_delta}, the control parameter $\delta$
is chosen to guarantee conditions (C1)-(C2). However, this ansatz is a crude
modification of the original perturbative model and could distort numerical
solutions.  A more subtle approach, possibly along the lines of flux-limited
diffusion, e.g. \cite{LevPom81}, would be preferable.}
\item{ 
An undesirable aspect of the RKDG method is the time step restriction required
by the explicit time integrator. For convection-diffusion equations, stiff
sources, and/or long time scales, this time step restriction is very
harsh. To lower the computational effort, (semi-)implicit time
discretizations are therefore necessary.  In addition, the method does not
address the challenges of spatially discontinuous fluxes that arise from
discontinuities in sources and material cross-sections.}
\item{
Another issue is the formation of unphysical shocks in the (perturbed) $M_1$
solution.  Further analysis is needed to understand why and when they appear
and how to further mitigate them.}
\end{itemize}

\appendix

\section{Appendix}

\begin{proof}[Computation of equation \eqref{eq:p_id}]
We first calculate two frequency integrals.  Let $\kappa:= \frac{h c}{k}
\ahat^T \bm$ and $\theta:= -\kappa \nu$.
Then
\begin{multline}
     \fint{\etad' \! \left(\frac{-h\nu c}{k} \ahat^T \bm \right)} {\nu}
     = \fint{ \frac{2h\nu^3 }{c^2} \frac{1}{\exp(-\kappa \nu)-1}}{\nu} \\
     =  - \frac{2h}{c^2 \kappa^4}
\fint{\frac{\theta^3}{\exp(\theta)-1}} {\theta}
    =  -\frac{2 \pi^4h}{15 c^2\kappa^4}
\end{multline}
and
\begin{multline}
     \fint{\etad'' \!\left(\frac{-h\nu c}{k}
\ahat^T  \bm\right)}{\nu}
     = \fint{ \frac{2h^2 \nu^4 }{kc}
\frac{\exp(-\kappa \nu)}{[\exp(-\kappa \nu)-1]^2}}{\nu} \\
     =  \frac{2h^2}{kc\kappa^5}
\fint{\frac{\theta^4\exp(\theta)}{[\exp(\theta)-1]^2}} {\theta}
    =  \frac{8\pi^4h^2}{15 kc\kappa^5}.
\end{multline}

With these two integrals, it is easy to show that $-\ahat/4$ is the unique
solution to the linear system
\begin{equation}
\Vint{\bm \bm^T \etad'' \!\left(\frac{-h\nu c}{k} \ahat^T  \bm\right)}
\bsbeta
    = \Vint{\bm \etad' \!\left(\frac{-h\nu c}{k} \ahat^T  \bm\right)},
\end{equation}
so that
\begin{equation}
\Vint{\bm \bm^T \etad'' \!\left(\frac{-h\nu c}{k} \ahat^T  \bm\right)}^{-1}
\Vint{\bm \etad' \!\left(\frac{-h\nu c}{k} \ahat^T  \bm\right)}
= -\frac{\ahat}{4}.
\end{equation}
Using the definition of $\P$,
\begin{equation}
\P \cE(\bu) = \frac{\ahat^T \bm}{4}\etad'' \!\left(\frac{-h\nu c}{k}
\ahat^T  \bm\right),
\end{equation}
and again the integrals above:
\begin{multline}
\fint{\P \cE(\bu)}{\nu} = \frac{-\ahat^T \bm}{4}\frac{8\pi^4h^2}{15
kc\kappa^5}
= \frac{-k\kappa}{4hc}\frac{8\pi^4h^2}{15
kc\kappa^5} \\
= -\frac{2 \pi^4h}{15 c^2 \kappa^4}
= \fint{\etad' \!\left(\frac{-h\nu c}{k} \ahat^T  \bm\right)}{\nu}
= \fint{\cE(\bu)}{\nu}
\end{multline}
\end{proof}

%
%
\begin{proof}[Proof of Lemma \ref{lem:pd_pc}]
The proof is a straight-forward calculation. It turns out to be more efficient
to calculate $\pd$ and $\pc$ without
directly using \eqref{eq:E_prime}.  Instead for any function $g$, we compute
\begin{equation}
 \vint{(\Omega \vee \Omega) \Pt g}
=  \vint{(\Omega \vee \Omega) g} - \sum_{k=0}^1 \frac{\p \Vint{(\Omega \vee \Omega)
\cE}}{\p\bu_k}  \vint{\bm_k g}
=  \vint{(\Omega \vee \Omega)g} - \sum_{k=0}^1 \frac{ \p \pe}{\p \bu_k}  \vint{\bm_k
g}. \label{eq:calc}
\end{equation}
Using \eqref{eq:calc}, we find for the diffusive
correction,
\begin{align}
\pd &= \frac{1}{c \sigt}\Vint{(\Omega \vee \Omega)\psid}  \nonumber \\
 &= -\frac{1}{c \sigt}\Vint{(\Omega \vee \Omega)\Pt (\Omega \cdot \grad_x
\cE(\bu)) }
\nonumber \\
  &= -\frac{1}{c \sigt}\div \Vint{(\Omega^{\vee 3}\cE(\bu)}
 +  \frac{1}{c \sigt}\sum_{k=0}^{1}\frac{ \p \pe}{\p \bu_k}  \div \vint{\bm_k
\Omega  \cE}
\nonumber \\
&=  -\frac{1}{c \sigt} \div \qe
 +  \frac{1}{c \sigt} \frac{ \p \pe}{\p E}  (\div F)
 +  \frac{1}{\sigt} \frac{ \p \pe}{\p F }  (\div \pe)
\end{align}
For the convection correction, we use the fact that $\phi$, $B$ and $S$ are
independent of $\Omega$.  This implies that $\vint{\bm_1 \phi} \equiv
\vint{\Omega \phi} = 0$ and similarly for $B$ and $S$.  We also use the Stefan
Boltzmann-Law \eqref{eq:SB_law} and the identity $\bu_0 = cE = \int_0^\infty
\phi \,
d\nu$. This gives
\begin{align}
\pc &= \frac{1}{c}\Vint{(\Omega \vee \Omega)\psic} \nonumber \\
  &= \frac{\rs}{4\pi c}
  \Vint{(\Omega \vee \Omega)\Pt \phi}
  + \frac{\ra}{4\pi c}
  \Vint{(\Omega \vee \Omega)\Pt B}
    + \frac{1}{4\pi c \sigt}
  \Vint{(\Omega \vee \Omega)\Pt S}
\nonumber \\
 & = \frac{\rs}{4\pi c}
  \left(\vint{(\Omega \vee \Omega)\phi} - \frac{ \p \pe}{\p E} \vint{\phi}
\right)
+ \frac{\ra}{4\pi c} \left(\vint{(\Omega \vee \Omega)B}- \frac{ \p \pe}{\p
E} \vint{B} \right)  \nonumber \\
& \qquad
+ \frac{1}{4\pi c \sigt} \left(\vint{(\Omega \vee \Omega)S}- \frac{ \p \pe}{\p
E} \vint{S}
\right) \nonumber \\
  &= \left( \third \Id -  \frac{ \p \pe}{\p E} \right)
  \left( \rs E + \ra a T^4  + \frac{S}{c \sigt} \right).
\end{align}
\end{proof}

%
%

\begin{proof}[Proof of Lemma \ref{lem:pe_qe}]
Let $\{\be_1, \be_2, \be_3\}$ be any orthogonal basis for $\bbR^3$.  Then
\begin{equation}
\Omega = \sum_{i=1}^3 \Omega_i \be_i \,, \quad \Omega_i := (\Omega \cdot \be_i)
\,, \quad \sum_{i=1}^3 \Omega_i^2 = 1, \label{eq:omega_expand}
\end{equation}
and
\begin{equation}
\Vint{(\Omega^{\vee k})\E} = \Vint{\left (\sum_{i=1}^3 \Omega_i
e_i\right)^{\vee k}\E} .
\end{equation}
Now set $\be_3 = \bn=F/|F|$ and note that, according to Lemma \ref{lem:colinear}
below, $\E$ depends on $\Omega$ only through $\Omega_3$. Thus, only the terms
with even
powers of $\Omega_1$ and $\Omega_2$
will survive.  For $k=2$, this means
\begin{equation} 
c \pe
= \Vint{\Omega_1^2\E}\be_1 \vee \be_1
 +\Vint{\Omega_2^2\E}\be_2 \vee \be_2
+\Vint{\Omega_3^2\E}\bn \vee \bn, \label{eq:pe_calc}
\end{equation}
and for $k=3$,
\begin{equation} \label{eq:qe_calc}
\qe
= 3\Vint{\Omega_1^2 \Omega_3 \E}\be_1 \vee \be_1 \vee \bn
 + 3\Vint{\Omega_2^2 \Omega_3 \E}\be_2 \vee \be_2 \vee \bn
 + \Vint{\Omega_3^3\E}\bn^{\vee 3}.
\end{equation}
The goal then is to write these formulas in terms of $\Omega_3$ only.  Let us
focus first on $\pe$. Because $\E$ depends only on $\Omega_3$, symmetry
arguments can be used to
conclude that first two terms in \eqref{eq:pe_calc} are the same.  Combined
with the far right relation \eqref{eq:omega_expand}, this gives
\begin{align}
c \pe
&= \Vint{\Omega_1^2\E}(\be_1 \vee \be_1 + \be_2 \vee \be_2)
+\Vint{(\Omega_3^2\E}\bn \vee \bn \nonumber \\
&= \Vint{\Omega_1^2\E}(\be_1 \vee \be_1 + \be_2 \vee \be_2 + \bn \vee \bn)
+\Vint{(\Omega_3^2-\Omega_1^2)\E} \bn \vee \bn \nonumber \\
&= \half \Vint{(1-\Omega_3^2)\E} \tx{Id}
+\half \Vint{(3\Omega_3^2-1)\E} \bn \vee \bn ,
\end{align}
where we have used the fact that $\be_1 \vee \be_1 + \be_2 \vee \be_2 + \bn \vee
\bn$ is the identity.  From the definition of $\chi_2$, we conclude that
\begin{equation}
\pe = \frac{E}{2}\left[(1- \chi_2) \Id + (3\chi_2-1)(\bn \vee \bn)\right]. \\
\end{equation}
Similarly for $k=3$,
\begin{align}
\qe
&= 3\Vint{\Omega_1^2\Omega_3\E}(\be_1 \vee \be_1 + \be_2 \vee \be_2 + \bn \vee
\bn) \vee
\bn
+\Vint{(\Omega_3^2-3\Omega_1^2)\Omega_3\E} \bn^{\vee 3} \nonumber \\
&= \frac{3}{2} \Vint{(1-\Omega_3^2)\Omega_3\E} \Id \vee \bn
+ \half \Vint{(5\Omega_3^2-3)\Omega_3\E} \bn^{\vee 3} \nonumber \\
& = \frac{3cE}{2}\left[(\chi_1 - \chi_3) (\Id \vee \bn) + (5 \chi_3 - 3\chi_1)
\bn^{\vee 3}\right].
\end{align}
%
%
%

\end{proof}

\begin{lem} \label{lem:colinear}
For the $M_1$ model, the multiplier $\hat{\alpha}_1$ is co-linear with $F$,
that is
\begin{equation}
    \frac{\hat{\alpha}_1}{|\hat{\alpha}_1|} = \frac{F}{|F|}
\end{equation}
\end{lem}

\begin{proof}
If $\E = \etad'\left( -\frac{h\nu c}{k} (\hat{\alpha}_0 + \hat{\alpha}_1\bm_1) \right)$ solves the optimization
problem \eqref{eq:optimization}, then by definition
\begin{equation}
F  = \Vint{\Omega \, \etad' \left(-\frac{h\nu c}{k} (\hat{\alpha}_0 +
\hat{\alpha}_1\bm_1) \right )}. \label{eq:M1_consistency}
\end{equation}
Let $R$ be any orthogonal $3 \times 3$ matrix which preserves $F$.  Then
multiplying
\eqref{eq:M1_consistency} by $R$ gives
\begin{equation}
     F  = \Vint{R \Omega \, \etad' \left(-\frac{h\nu c}{k} (\hat{\alpha}_0 +
\hat{\alpha}_1\bm_1) \right )}
     = \Vint{\Omega \, \etad' \left( -\frac{h\nu c}{k} (\hat{\alpha}_0 + R
\hat{\alpha}_1\bm_1) \right)},
\end{equation}
where we have used the fact that the measure $d\Omega$ is invariant under the
action of $R$.  Because the solution of the optimization is unique, we conclude
that $R\hat{\alpha}_1 = \hat{\alpha}_1$ and therefore, since $R$ is arbitrary, $\hat{\alpha}_1$
and $F$ must be co-linear   hjhkjhkj
\end{proof}

%
%

\begin{proof}[Proof of Proposition \ref{prop:pM1_hyperbolic}]
Without loss of generality, we consider $c=1$ and prove that the eigenvalues of the Jacobian associated with the convective flux in \eqref{eq:conv_flux}
are real.
To do so, the following definitions are introduced:
\begin{align*}
\alpha &:= \frac{\p}{\p E} \left (\xi E +r_a a T^4 +\frac{S}{\sig{t}}  \right ),  & \beta &:= \frac{\p}{\p F} \left (\xi E +r_a a T^4 +\frac{S}{\sig{t}}  \right ), \\
\xi(f) &:= \chi(f) +r_s\eta(f), &  f&:=F/E.
\end{align*} 
We show that the radical $\alpha +\beta^2/4$ in the formula for the eigenvalues is positive for all $f\neq 1$.
Note that \eqref{eq:eta_abstract} implies $\eta=1/3 +\chi' f -\chi$ and hence,
\begin{align}
\xi = r_s \left (\frac13 -\chi +\chi' f \right) +\chi.
\end{align}
The prime notation always refers to the derivative with respect to $f$. 
With this, we conclude
\begin{align}
\beta^2 +4\alpha &= \xi'^2 -4f \xi' +4\xi = \xi'^2 -4f \xi' + 4 r_s \left (\frac13 -\chi +\chi' f \right ) + 4\chi \\
		&= (\xi' -2f)^2 +4 r_s \left (\frac13 -\chi +\chi' f \right) +4 (\chi -f^2). \label{eq:radical}
\end{align}
Using \eqref{eq:chi_eta}, straight-forward calculations imply
\begin{align}
\chi -f^2 > 0 \q \tx{for all } f\neq 1 \quand  \frac13 -\chi +\chi' f \geq 0. \label{eq:chi_ineq}
\end{align}
Applying \eqref{eq:chi_ineq} on \eqref{eq:radical} completes the proof.

\end{proof}

\bibliographystyle{siam}
\bibliography{literature,entropy_closures}
\end{document}